\documentclass[runningheads]{llncs}

\usepackage[utf8]{inputenc}

\usepackage{comment}

\usepackage{tcolorbox}
\usepackage{amsmath}
\usepackage{enumerate}
\usepackage{graphicx}
\usepackage{amsmath}
\usepackage{amssymb}
\usepackage{comment}
\usepackage{hyperref}
\usepackage[english]{babel}

\usepackage{tikz}
\usetikzlibrary{shapes.geometric}
\usepackage{float}
\usepackage{tcolorbox}
\usepackage{graphicx}

\newtheorem{observation}{Observation}[section]

\usepackage{import}

\def \calp {{\cal P}}
\def \calf {{\cal F}}

\def \psp {{{{\textsc{Path Set Packing}}}}}
\def \psps {{{{\textsc{PSP}}}}}
\def \sp {{{{\textsc{Set Packing}}}}}
\def \mis {{{{\textsc{Max Independent Set}}}}}

\def \is {{{{\textsc{Independent Set}}}}}

\def \kmcc {{{{\textsc{$k$-Multi Colored Clique}}}}}
\def \kmccs {{{{\textsc{$k$-MCC}}}}}

\begin{document}

\title{Parameterized Complexity of Path Set Packing}
\titlerunning{Path Set Packing}
%
%
\author{N.R. Aravind\and
Roopam Saxena* }
\authorrunning{Aravind and Saxena}
%
\institute{Department of Computer Science and Engineering\\ IIT Hyderabad, Hyderabad, India\\
\email{aravind@cse.iith.ac.in, cs18resch11004@iith.ac.in*}}
\maketitle              

\begin{abstract}
In {\psp}, the input is an undirected graph $G$, a collection $\calp$ of simple paths in $G$, and a positive integer $k$. The problem is to decide whether there exist $k$ edge-disjoint paths in $\calp$. We study the parameterized complexity of {\psp} with respect to both natural and structural parameters. We show that the problem is $W[1]$-hard with respect to vertex cover number, and $W[1]$-hard respect to pathwidth plus maximum degree plus solution size. These results answer an open question raised in \cite{XZ2018}.
On the positive side, we present an FPT algorithm parameterized by feedback vertex number plus maximum degree, and present an FPT algorithm parameterized by treewidth plus maximum degree plus maximum length of a path in $\calp$. These positive results complement the hardness of {\psp} with respect to any subset of the parameters used in the FPT algorithms. We also give a $4$-approximation algorithm for maximum path set packing problem which runs in FPT time when parameterized by feedback edge number.

\end{abstract}

\section{Introduction}
\subsection{Problem Definition and Previous Work}
The path set packing problem was introduced by Xu and Zhang \cite{XZ2018}. A collection of simple paths $P$ in a graph $G$ is a path set packing if all the paths in $P$ are pairwise edge disjoint. We now formally define {\psp}.
\begin{tcolorbox}[colback=white]
{
{\psp} (\psps): \newline
\textit{Input:} An instance $I$ = $(G,\calp, k)$, where $G=(V,E)$ is an undirected graph, $\calp$ is a collection of simple paths in $G$, and $k\in \mathbb{N}$.\newline
\textit{Output:} YES, if there exists a $S\subseteq \calp$ such that $|S|\geq k$ and $S$ is a path set packing; NO otherwise.
}
\end{tcolorbox} 
\paragraph{}

Xu and Zhang \cite{XZ2018} showed that {\psp} is NP-complete even when the maximum length of the given paths is no more than $3$. Considering the optimization version of the problem, they showed that finding maximum path set packing is hard to approximate within $O(|E|^{{1\over{2}}-\epsilon})$ unless NP=ZPP. They showed that {\psp} can be solved in polynomial time when the input graph is a tree. They also gave a parameterized algorithm for maximum path set packing problem with running time $O(|\mathcal{P}|^{tw(G)\Delta}|V|$) where $tw(G)$ is treewidth of $G$ and $\Delta$ is maximum degree. Further, they left open the question whether {\psp} is fixed parameter tractable with respect to treewidth of the input graph.

\subsection{Related Work}
Given a universe $\cal U$ and a family $\cal F$ of subsets of $\cal U$, {\sp} is the problem  of deciding if there exists a subfamily $F\subseteq \cal F$ of size at least $k$ such that all the sets in $F$ are pairwise disjoint. {\sp} is known to be W[1]-hard when parameterized by solution size $k$ \cite{DBLP:series/mcs/DowneyF99}. When the maximum size of a set in $\calf$ is $d$, FPT algorithms for the combined parameter of $k$ and $d$ have been obtained \cite{JIA2004106,KOUTIS20057}. Kernels of size $O(k^d)$ \cite{Fellowsmichael2004} and $O(k^{d-1})$ \cite{ABUKHZAM2010621} have also been obtained. Since {\psp} can be seen as a special case of {\sp}, all the positive results obtained for {\sp} are also applicable on {\psp}. Thus, {\psp} is FPT when parameterized by maximum length of a path in $\cal P$ plus solution size $k$. Since length of a path in a graph $G$ is bounded by its vertex cover number, {\psp} is FPT when parameterized by the vertex cover number plus solution size.
{\sp} has also been studied extensively in the realm of approximation algorithms \cite{halldorsson1995approximating,halld1998}.
\paragraph{}

Finding a maximum size path set packing can also be seen as the problem of finding a maximum independent set on the {\it conflict graph} obtained by considering each path as a vertex with two vertices being adjacent if the corresponding paths share an edge.
Like {\sp}, the {\mis} (MIS) problem and its decision version {\is} problem have been extensively studied; of particular relevance is the study of MIS for intersection graphs, i.e. graphs whose vertices represent mathematical objects with edges representing objects that intersect.
For example, $O(n^{\varepsilon})$ approximation algorithms were obtained in \cite{FP11} for intersection graphs of curves on the plane (also called string graphs).
In \cite{KT14}, the authors introduced measures of similarity to the class of chordal graphs, and as a consequence obtained approximation algorithms for intersection graphs of various geometric objects. Their paper also includes a tabular summary of such results.

\paragraph{}

{\psp} itself has been mainly studied for the case when the underlying graph $G$ is a grid graph; the corresponding conflict graphs are called EPG graphs \cite{GolumbicLS09}. It was shown in \cite{GolumbicLS09} that every graph is an EPG graph. Further, there is no $2^{o(n)}$ time exact algorithm for MIS assuming ETH \cite{DBLP:journals/eatcs/LokshtanovMS11}, {\is} is W[1]-hard by solution size \cite{DBLP:series/mcs/DowneyF99}, and
these results immediately implies the following.

\begin{corollary}\label{theorem:planar-hardness}
 {\psp} is W[1]-hard on Grid graphs when parameterized by solution size $k$.
\end{corollary}

\begin{corollary}\label{exact-lowerbound}
The problem of finding maximum path set packing does not admit a $2^{o(|\calp|)}$ time exact algorithm, assuming ETH.
\end{corollary}

\paragraph{}

Thus, it is natural to consider {\psp} with further or different restrictions on the input graph $G$.
We mention some of the known results where the input graph is a tree or a grid, and with some restrictions on the paths are imposed.
\begin{enumerate}
\item When $G$ is a tree, the conflict graph is called an EPT graph \cite{GolumbicJ85_EPT}, recognizing EPT graph is NP-Complete \cite{GolumbicJ85_EPT_NP}. {\mis} is solvable in polynomial time on the class of EPT graphs \cite{TARJAN1985221}. 
\item The class of $B_k$-EPG graphs was defined as graphs obtained as the edge intersection graph of paths on a grid, with the restriction that each path have at most $k$ bends.
Recognizing $B_1$-EPG graphs is NP-hard  \cite{HELDT2014144}. Further, {\mis} on $B_1$-EPG graphs is NP-hard \cite{EGM13} . 
\item Approximation algorithms: In \cite{EGM13}, the authors show that {\mis} on the class of $B_1$-EPG graphs admits a 4-approximation algorithm. 
\item Fixed Parameter Tractability: In \cite{BBCGP20}, the authors showed that for the class of $B_1$-EPG graphs, when the number of path shapes is restricted to three, {\is} admits an FPT algorithm parameterized by solution size, while remanining W[1]-hard on $B_2$-EPG graphs with parameter solution size.
\end{enumerate}

\subsection{Our Results}

We studied {\psp} with respect to combination of both natural and structural parameters of the input graph $G$, specifically the structural parameters of the input graph $G$.
We obtained the following hardness results for {\psp}.

\begin{theorem}\label{theorem:vc-hardness}
{\psp} is W[1]-hard when parameterized by vertex cover number of input graph $G$.
\end{theorem}
We note that the maximum length of a simple path in $G$ is bounded by its vertex cover number. Thus, {\psp} remains W[1]-hard when parameterizeed by vertex cover number plus maximum length of a path in $\cal P$.
\begin{theorem}\label{theorem:Pathwidth-hardness}
{\psp} on graphs of degree at most $4$ is W[1]-hard when parameterized by pathwidth of input graph $+$ solution size.
\end{theorem}

On the positive side, we obtained the following parameterized algorithms for {\psp}, which are summarized in the following theorems.
\begin{theorem}\label{theorem:FVD_maxdegree}
{\psp} admits an FPT algorithm when parameterized by feedback vertex number of $G$ $+$ maximum degree of $G$.
\end{theorem}

\begin{theorem}\label{theorem:4-apxFEN}
There exists a $4$-approximation algorithm for finding maximum size path set packing in $\calp$ which runs in FPT time when parameterized by feedback edge number  of input graph $G$.
\end{theorem}

\begin{theorem}\label{theorem:tw_maxdegree}
{\psp} admits an FPT algorithm when parameterized by  treewidth of $G$ $+$ maximum degree of $G$ $+$ maximum length of a path in $\calp$.
\end{theorem}

We note that our positive results complement the hardness of {\psp} with respect to any subset of the parameters used in the respective algorithms.  
For maximum degree plus maximum path length, we note that the reduction from {\mis} to maximum path set packing problem given in \cite{XZ2018} to prove the inapproximability of maximum path set packing problem also works to prove NP-hardness of {\psp} for bounded maximum degree in $G$ and bounded maximum length of a path in $\mathcal{P}$ using the fact that {\is} is NP-hard on bounded degree graphs \cite{garey}.
\section{Preliminaries}

\subsection{Sets and Sequences.}
We use $[n]$ to denote the set $\{1,2,....,n\}$. A sequence is a list of elements in a particular order. For a sequence $\rho$, its length is denoted by $|\rho|$ and it is the number of elements in $\rho$. For $j\in[|\rho|]$, $\rho[j]$ is the element of $\rho$ at index $j$. $set(\rho)$ denotes the set formed by all the elements of $\rho$. 

For a set $X$, a collection $P$ of subsets of $X$ is a partition of $X$ if
\begin{itemize}
    \item sets in $P$ are pairwise disjoint,
    \item $P$ does not contain an empty set, and
    \item $X= \bigcup_{p\in P} p$.
\end{itemize}

\subsection{Graphs}
 All the graphs consider in this paper are simple and finite. We use standard graph notations and terminologies and refer the reader to \cite{DBLP:books/daglib/0030488} for basic graph notations and terminologies. We mention some notations used in this paper. A graph  $G=(V,E)$ has a vertex set $V$ and edge set $E$. We also use $V(G)$ and $E(G)$ to denote the vertex set and edge set of $G$ respectively. 
 For an edge set $F\subseteq E(G)$, $V(F)$ denotes the set of all the vertices of $G$ with at least one edge in $F$ incident on it. For two disjoint vertex sets $A,B\subseteq V$, we use $E_G(A,B)$ to denote all the edges in graph $G$ with one endpoint in $A$ and other in $B$, if the graph in context is clear then we simply use $E(A,B)$. For a vertex set $S\subseteq V(G)$, $G[S]$ denotes the induced sub graph of $G$ on vertex set $S$, and $G-S$ denotes the graph $G[V(G)\setminus S]$. For an edge set $A\subseteq E(G)$, $G-A$ denotes the graph with vertex set $V(G)$ and edge set $E(G)\setminus A$.
 A component $C$ of a graph $G$ is a maximally connected subgraph of $G$.
 A graph $G$ is a forest if every component of $G$ is a tree.
\paragraph{}
 \textbf{Tree Decomposition} \cite{DBLP:books/sp/CyganFKLMPPS15} : A tree decomposition of a graph $G$ is a pair $(T,\beta)$ where $T$ is a tree and $\beta$ (called a bag) is a mapping that assigns to every $t \in V (T)$ a set $\beta(t)\subseteq V (G)$, such that the following
holds: 
\begin{enumerate}
    \item For every $e \in E(G)$, there exists a $ t\in V(T)$ such that $V(e) \subseteq \beta (t);$
    \item For $v\in V(G)$, let $\beta ^{-1}(v)$ be the set of all vertices $t \in V(T)$ such that $v \in \beta (t)$, then  $T[\beta ^{-1}(v)]$ is a connected nonempty subgraph of $T$.
\end{enumerate}
If the tree $T$ is rooted at some node $r$, we call it a \textit{rooted tree decomposition}. If $T$ is a path then it is called a path decomposition. 
The width of the tree decomposition $(T,\beta)$ is the $\max \{|\beta(t)|-1\mid t\in T\}$. Treewidth of $G$ is defined to be the minimum width of any tree decomposition of $G$. The pathwidth of $G$ is similarly defined using path decompositions of $G$.
 \paragraph{}
 A vertex set $S\subseteq V(G)$ is a vertex cover of $G$ if the sub graph $G[V\setminus S]$ has no edge. The minimum size of any vertex cover of $G$ is called vertex cover number of $G$. A vertex set $S\subseteq V(G)$ is a feedback vertex set of $G$ if the sub graph $G[V\setminus S]$ has no cycle. The minimum size of any feedback vertex set of $G$ is called feedback vertex number (FVN) of $G$. An edge set $F\subseteq E(G)$ is a feedback edge set of $G$ if the sub graph $G-F$ has no cycle. The minimum size of any feedback edge set of $G$ is called feedback edge number (FEN) of $G$. For a connected graph $G$, its feedback edge number $\lambda = |E(G)|-|V(G)|+1$. The set of all the edges incident on a feedback vertex set forms a feedback edge set, if the graph has a feedback vertex number $\Gamma$, then $\lambda\leq \Gamma \cdot \Delta$.
 \paragraph{}
 A path in a graph $G$ is a sequence of distinct vertices such that successive vertices are connected by an edges. The first and last vertices of a path are its endpoints. We call a path simple to emphasize that all its vertices are distinct. We denote the set of all the edges and all the vertices of a path $p$ by $E(p)$ and $V(p)$ respectively. Similarly, for a collection of paths $P$, we use $E(P)$ to denote the set $\bigcup_{p\in P} E(p)$, and $V(P)$ to denote the set $\bigcup_{p\in P} V(p)$. A collection of paths $P$ is a path set packing if all the paths in $P$ are pairwise edge disjoint.

\subsection{Parameterized Complexity} 

For the details on parameterized complexity, we refer to \cite{DBLP:books/sp/CyganFKLMPPS15,DBLP:series/txcs/DowneyF13}, and recall some definitions here.

\begin{definition}[\cite{DBLP:books/sp/CyganFKLMPPS15}]
A \textit{parameterized problem} is a language $L \subseteq \Sigma^* \times \mathbb{N} $ where $\Sigma$ is a fixed and finite alphabet. For an instance $I=(x,k) \in \Sigma^* \times \mathbb{N} $, $k$ is called the parameter.  A parameterized problem is called \textit{fixed-parameter tractable} (FPT) if there exists an algorithm $\cal A$ (called a \textit{fixed-parameter algorithm} ), a computable function $f:\mathbb{N} \to \mathbb{N}$, and a constant $c$ such that, the algorithm $\cal A$ correctly decides whether $(x,k)\in L$ in time bounded by $f(k).|(x,k)|^c$. The complexity class containing all fixed-parameter tractable problems is called FPT.
\end{definition}

Informally, a W[1]-hard problem is unlikely to be fixed parameter tractable, see \cite{DBLP:books/sp/CyganFKLMPPS15} for details on complexity class W[1].
\begin{definition}[\cite{DBLP:books/sp/CyganFKLMPPS15}]
Let $P,Q$ be two parameterized problems. A parameterized reduction from $P$ to $Q$ is an algorithm which for an instance $(x,k)$ of $P$ outputs an instance $(x',k')$ of $Q$ such that:
\begin{itemize}
    \item $(x,k)$ is yes instance of $P$ if and only if $(x',k')$ is a yes instance of $Q$,
    \item $k'\leq g(k)$ for some computable function $g$, and
    \item the reduction algorithm takes time $f(k)\cdot |x|^{O(1)}$ for some computable function $f$
\end{itemize}
\end{definition}

\begin{theorem}[\cite{DBLP:books/sp/CyganFKLMPPS15}]\label{def-parameterized reduction}
If there is a parameterized reduction from $P$ to $Q$ and $Q$ is fixed parameter tractable then $P$ is also fixed parameter tractable.
\end{theorem}

\subsection{Problem Definitions}
A collection of sets $Q$ is a set packing if all the sets in $Q$ are pairwise disjoint.

\begin{tcolorbox}[colback=white]
{
$d$-{\sp}: \newline
\textit{Input:} An instance $I$ = $(U,{S},k)$, where  $U$ is a universe, $S$ is a collection of subsets of $U$, where each subset contains at most $d$ elements, and $k \in \mathbb{N}$.\newline
\textit{Output:}  A set packing of size $k$ in $S$, or report that no set packing of size $k$ exists in $S$. \newline
} 
\end{tcolorbox}

\begin{theorem}[\cite{WANG2008748}]\label{thm:3setpacking}
   $3$-{\sp} can be solved in time $2^{O(k)}\cdot (|{\cal Q}|)^{O(1)}$.
\end{theorem}

\section{Parameterized by Vertex Cover Number}

\begin{figure}[h]

    \centering
    \begin{tikzpicture}


    \foreach \i in {1,2}{
     
            \node[shape=circle, draw, fill = black, scale= 0.3, font=\footnotesize,label={\scriptsize{$x_{i,\i}$}}] ({\i+2}) at (-1+\i/1.2,-0.2+1){};
   }
   \foreach \i in {3,4}{
     
            \node[shape=circle, draw, fill = black, scale= 0.1, font=\footnotesize,] ({\i+2}) at (-1+\i/1.2,-0.2+1){};
   }
   \foreach \i in {5}{
     
            \node[shape=circle, draw, fill = black, scale= 0.3, font=\footnotesize,label={\scriptsize{$x_{i,n-1}$}}] ({\i+2}) at (-1+\i/1.2,-0.2+1){};
   }

    \foreach \i in {1}{
     
            \node[shape=circle, draw, fill = black, scale= 0.3, font=\footnotesize,label={180:\scriptsize{$c_{i,\i}$}}] ({\i}) at (0+\i/1.2,-1){};
   }
   \foreach \i in {2}{
     
            \node[shape=circle, draw, fill = black, scale= 0.3, font=\footnotesize,label={\scriptsize{$c_{i,\i}$}}] ({\i}) at (0+\i/1.2,-1){};
   }
   \foreach \i in {3,4}{
     
            \node[shape=circle, draw, fill = black, scale= 0.1, font=\footnotesize,] ({\i}) at (0+\i/1.2,-1){};
   }
    \foreach \i in {5}{
     
            \node[shape=circle, draw, fill = black, scale= 0.3, font=\footnotesize,label={\scriptsize{$c_{i,k}$}}] ({\i}) at (0+\i/1.2,-1){};
   }


    \foreach \i in {1,2}{
     
            \node[shape=circle, draw, fill = black, scale= 0.3, font=\footnotesize,label={270:\scriptsize{$v_{i,1,\i}$}}] ({\i+0}) at (-2+\i/1.2,-3){};
   }
   \foreach \i in {3,4}{
     
            \node[shape=circle, draw, fill = black, scale= 0.1, font=\footnotesize,] ({\i+0}) at (-1.6+\i/1.7,-3){};
   }
    \foreach \i in {5}{
     
            \node[shape=circle, draw, fill = black, scale= 0.3, font=\footnotesize,label={270:\scriptsize{$v_{i,1,k}$}}] ({\i+0}) at (-2+\i/1.4,-3){};
   }
   
    
    \foreach \i in {1,2}{
     
            \node[shape=circle, draw, fill = black, scale= 0.3, font=\footnotesize,label={270:\scriptsize{$v_{i,n,\i}$}}] ({\i+1}) at (3+\i/1.2,-3){};
   }
   \foreach \i in {3,4}{
     
            \node[shape=circle, draw, fill = black, scale= 0.1, font=\footnotesize,] ({\i+1}) at (3.4+\i/1.6,-3){};
   }
    \foreach \i in {5}{
     
            \node[shape=circle, draw, fill = black, scale= 0.3, font=\footnotesize,label={270:\scriptsize{$v_{i,n,k}$}}] ({\i+1}) at (3+\i/1.4,-3){};
   }

   \foreach \i in {1,2,...,5}{
            \draw[thin,black!15] (\i)--(\i+1);
            \draw[thin,black!15] (\i)--(\i+1);

   }
   \draw[thin,black!15] (2)-- (1+1) ;
   \draw[thin,black!15] (3)-- (2+1) ;
   \draw[thin,black!15] (4)-- (3+1) ;
   \draw[thin,black!15] (5)-- (4+1) ;

   \draw[thick,black] (1)--(2+2);
   \draw[thin,black!15] (1)--(1+2);
   \foreach \i in {3,...,5}{
            \draw[thin,black!15] (1)--(\i+2);
            \draw[thin,black!15] (1)--(\i+2);

   }

    \foreach \i in {1,2,...,5}{
            \draw[thick,black] (\i)--(\i+0);

   }
   \draw[thick,black] (2)-- (1+0) ;
   \draw[thick,black] (3)-- (2+0) ;
   \draw[thick,black] (4)-- (3+0) ;
   \draw[thick,black] (5)-- (4+0) ;
   
  

\end{tikzpicture}
\caption{An example of vertex selection gadget $H_i$, the darkened edges forms a long path $P_{x_{i,2},V_{i,1}}=(x_{i,2},c_{i,1},v_{i,1,1},c_{i,2},v_{i,1,2},......., c_{i,k},v_{i,1,k})$. }\label{fig:longpath}
\end{figure}

\begin{figure}[h]

    \centering
    \begin{tikzpicture}[scale=1]


 \node [ellipse,draw= black!10, minimum height=5.5cm,minimum width= 3cm, label ={90:$H_i$}] (z) at (-1.75,-1.2) {};
   
    \foreach \i in {1,3}{
     
            \node[shape=circle, draw, fill = black, scale= 0.2, font=\scriptsize] ({\i}) at (-1,-\i/2){};
   }
   \foreach \i in {2}{
     
            \node[shape=circle, draw, fill = black, scale= 0.4, font=\scriptsize,label={0:\scriptsize{$c_{i,l}$}}] ({\i}) at (-1,-\i/2){};
   }
   \foreach \i in {4}{
     
            \node[shape=circle, draw, fill = black, scale= 0.4, font=\scriptsize,label={-30:\scriptsize{$c_{i,j}$}}] ({\i}) at (-1,-\i/2){};
   }
  

    \foreach \i in {1,2,3,4}{
     
            \node[shape=circle, draw, fill = black, scale= 0.2, font=\scriptsize,label={}] ({\i+0}) at (-2,1.5-\i/3){};
   }
  
   
    
    \foreach \i in {1,2,3}{
     
            \node[shape=circle, draw, fill = black, scale= 0.2, font=\scriptsize,label={}] ({\i+1}) at (-2.5,-0.2-\i/3){};
   }
      \foreach \i in {4}{
     
            \node[shape=circle, draw, fill = black, scale= 0.4, font=\scriptsize,label={-90:\scriptsize{$v_{i,i',j}$}}] ({\i+1}) at (-2.5,-0.2 -\i/3){};
   }

  \foreach \i in {1,2,3,4}{
     
            \node[shape=circle, draw, fill = black, scale= 0.2, font=\scriptsize,label={}] ({\i+2}) at (-2,-2-\i/3){};
   }

    \foreach \i in {1,2,...,4}{
            \draw[thin,black!20] (\i)--(\i+0);

   }
   \draw[thin,black!20] (2)-- (1+0) ;
   \draw[thin,black!20] (3)-- (2+0) ;
   \draw[thin,black!20] (4)-- (3+0) ;
  
   \foreach \i in {1,2,...,4}{
            \draw[thin,black!20] (\i)--(\i+1);

   }
   \draw[thin,black!20] (2)-- (1+1) ;
   \draw[thin,black!20] (3)-- (2+1) ;
   \draw[thin,black!20] (4)-- (3+1) ;

   \foreach \i in {1,2,...,4}{
            \draw[thin,black!20] (\i)--(\i+2);

   }
   \draw[thin,black!20] (2)-- (1+2) ;
   \draw[thin,black!20] (3)-- (2+2) ;
   \draw[thin,black!20] (4)-- (3+2) ;

  


  \node [ellipse,draw= black!10, minimum height=5.5cm,minimum width= 3cm, label ={90:$H_j$}] (z) at (4.75,-1.2) {};
  
    \foreach \i in {1,4}{
     
            \node[shape=circle, draw, fill = black, scale= 0.2, font=\footnotesize] ({\i+10}) at (4,-\i/2){};
   }
   \foreach \i in {2}{
     
            \node[shape=circle, draw, fill = black, scale= 0.4, font=\footnotesize,label={180:\scriptsize{$c_{j,l}$}}] ({\i+10}) at (4,-\i/2){};
   }
   \foreach \i in {3}{
     
            \node[shape=circle, draw, fill = black, scale= 0.4, font=\footnotesize,label={230:\scriptsize{$c_{j,i}$}}] ({\i+10}) at (4,-\i/2){};
   }
  

    \foreach \i in {1,2,3,4}{
     
            \node[shape=circle, draw, fill = black, scale= 0.2, font=\footnotesize,label={}] ({\i+10+0}) at (5,1.5-\i/3){};
   }
  
   
    
    \foreach \i in {1,2,4}{
     
            \node[shape=circle, draw, fill = black, scale= 0.2, font=\footnotesize,label={}] ({\i+10+1}) at (5.5,-0.2 -\i/3){};
   }
      \foreach \i in {3}{
     
            \node[shape=circle, draw, fill = black, scale= 0.4, font=\footnotesize,label={-80:\scriptsize{$v_{j,j',i}$}}] ({\i+10+1}) at (5.5,-0.2 -\i/3){};
   }

  \foreach \i in {1,2,3,4}{
     
            \node[shape=circle, draw, fill = black, scale= 0.2, font=\footnotesize,label={}] ({\i+10+2}) at (5,-2 -\i/3){};
   }

    \foreach \i in {1,2,...,4}{
            \draw[thin,black!20] (\i+10)--(\i+10+0);

   }
   \draw[thin,black!20] (2+10)-- (1+10+0) ;
   \draw[thin,black!20] (3+10)-- (2+10+0) ;
   \draw[thin,black!20] (4+10)-- (3+10+0) ;

   \foreach \i in {1,2,...,4}{
            \draw[thin,black!20] (\i+10)--(\i+10+1);

   }
   \draw[thin,black!20] (2+10)-- (1+10+1) ;
   \draw[thin,black!20] (3+10)-- (2+10+1) ;
   \draw[thin,black!20] (4+10)-- (3+10+1) ;

   \foreach \i in {1,2,...,4}{
            \draw[thin,black!20] (\i+10)--(\i+10+2);

   }
   \draw[thin,black!20] (2+10)-- (1+10+2) ;
   \draw[thin,black!20] (3+10)-- (2+10+2) ;
   \draw[thin,black!20] (4+10)-- (3+10+2) ;


  \node [ellipse,draw= black!10, minimum height=1.8cm,minimum width= 2.5cm, label ={30:$H_l$}] (z) at (1.5,0.35) {};
  
    \foreach \i in {2,5}{
     
            \node[shape=circle, draw, fill = black, scale= 0.2, font=\footnotesize] ({\i+20}) at (0.3+\i/3,0){};
   }
   \foreach \i in {3}{
     
            \node[shape=circle, draw, fill = black, scale= 0.4, font=\footnotesize,label={-90:\scriptsize{$c_{l,i}$}}] ({\i+20}) at (0.3+\i/3,0){};
   }
  
  \foreach \i in {4}{
     
            \node[shape=circle, draw, fill = black, scale= 0.4, font=\footnotesize,label={-87:\scriptsize{$c_{l,j}$}}] ({\i+20}) at (0.3+\i/3,0){};
   }

    \foreach \i in {2,3,4,5}{
     
            \node[shape=circle, draw, fill = black, scale= 0.2, font=\footnotesize,label={}] ({\i+20+0}) at (0.5+\i/4,1){};
   }
  
   

    \foreach \i in {2,...,5}{
            \draw[thin,black!20] (\i+20)--(\i+20+0);

   }
   
   \draw[thin,black!20] (3+20)-- (2+20+0) ;
   \draw[thin,black!20] (4+20)-- (3+20+0) ;
   \draw[thin,black!20] (5+20)-- (4+20+0) ;

  \draw[thick,black!20] (2)-- (3+20) ;
   \draw[thick,black!60] (4)-- (3+10) ;
   \draw[thick,black!20] (2+10)-- (4+20) ;
   
   \draw[thick,black!60] (3+10)-- (3+10+1) ;
   \draw[thick,black!60] (4)-- (4+1) ;

\end{tikzpicture}
\caption{ An example of inter gadget edges, and darkened edges forms a short path  $P_{v_{i,i',j},v_{j,j',i}}$ corresponding to an edge $v_{i,i'}v_{j,j'}$ in $G$. }\label{fig:shortpath}
\end{figure}

In the {\kmcc} ({\kmccs}) problem we are given a graph $G=(V,E)$, where $V$ is partitioned into $k$ disjoint sets $V_1,...,V_k$, each of size $n$, and the question is if $G$ has a clique $\mathcal{C}$ of size $k$ such that $|\mathcal{C}\cap V_i|=1$ for every $i\in [k]$. It is known that {\kmccs} is $W[1]$-hard parameterized by $k$ \cite{FHRV09}.

We will give a parameterized reduction from {\kmccs} to {\psp}. Let $G=(V,E)$ and $\{V_1,...V_k\}$ be an input of {\kmccs}. For the simplicity of notations, let the vertices of every set $V_i$ are labeled  $v_{i,1}$ to $v_{i,n}$. We will construct an equivalent instance $(G'=(V',E'), \mathcal{P},k')$ of {\psp} (see Figure  \ref{fig:longpath} and Figure  \ref{fig:shortpath} for overview). For every set $V_i$, we construct a vertex selection gadget $H_i$ (an induced subgraph of $G'$) as follows .
\begin{itemize}
    \item Create a set $C_i=\{c_{i,1},...c_{i,k}\}$ of $k$ vertices, a set $X_i=\{x_{i,1},x_{i,2},...,x_{i,n-1}\}$ of $n-1$ vertices, and connect every $x_{i,j}$ to $c_{i,1}$ where $j\in [n-1]$, we call these edges $E_{C_i,X_i}$. 
    \item For every $v_{i,j}\in V_i$, create a vertex set $V_{i,j}= \{v_{i,j,1},v_{i,j,2},...,v_{i,j,k}\}$ of $k$ vertices, connect $v_{i,j,l}$ to $c_{i,l}$ and $c_{i,l+1}$ for every $l\in[k-1]$, and connect $v_{i,j,k}$ to $c_{i,k}$. We denote these edges by $E_{C_i,V_{i,j}}$.
\end{itemize}

Formally, $H_i= (\bigcup_{j=1}^n V_{i,j} \cup X_i\cup C_i ,  \bigcup_{j=1}^n E_{C_i,V_{i,j}} \cup E_{C_i,X_i})$. Further, let $C=\bigcup_{i=1}^k C_i$,
we add the following edges in $G'$.

\begin{itemize}
    \item For $1\leq i<j \leq k$, we connect $c_{i,j}\in C_i$ to $c_{j,i}\in C_j$. We call these edges the inter gadget edges and denote them by $E_C$. Observe that there are $k \choose 2$ inter gadget edges (Figure  \ref{fig:shortpath}).
\end{itemize}

The above completes the construction of $G' = (\bigcup_{i=1}^k V(H_i), \bigcup_{i=1}^k E(H_i) \cup E_C)$.

We now move on to the construction of the collection $\mathcal{P}$.
\begin{itemize}
    \item Let $P_{x_{i,l},V_{i,j}}= (x_{i,l},c_{i,1},v_{i,j,1},c_{i,2},v_{i,j,2},......., c_{i,k},v_{i,j,k})$, that is a path starting at $x_{i,l}$ and then alternatively going through a vertex in $C_i$ and $V_{i,j}$ and ending at $v_{i,j,k}$. We call such a path a long path. For every $i\in [k]$, $l\in [n-1]$, and $j\in[n]$ we add $P_{x_{i,l},V_{i,j}}$ in $\mathcal{P}$. Observe that there are $n(n-1)$ long paths added from every $H_i$.
    \item For every edge $e=v_{i,i'}v_{j,j'} \in E$ where $v_{i,i'}\in V_i$ and $v_{j,j'}\in V_j$, w.l.o.g. assuming $i<j$, we add a path $P_{v_{i,i',j},v_{j,j',i}}=\{v_{i,i',j},c_{i,j},c_{j,i},v_{j,j',i}\}$ in $\mathcal{P}$ and call such a path, a short path. There are $|E|$ short paths added to $\mathcal{P}$.
\end{itemize}
We set  $k'=k(n-1)+$ $k\choose 2$. The above completes the construction of instance $(G'=(V',E'),\mathcal{P},k')$ with $|\mathcal{P}|= |E|+kn(n-1)$. 

\paragraph{}

The vertex set $C$ forms a vertex cover for $G'$ which is of size $k^2$, and the length of every long path is $2k+1$.  Further, the time taken for construction is $(|V|)^{O(1)}$, which ensures that the reduction is a parameterized reduction. The following concludes the correctness of the reduction and proof of Theorem \ref{theorem:vc-hardness}.

\begin{lemma}\label{lemma:vertexcovercorrectness}
$(G=(V,E),\{V_1,...V_k\})$  is a yes instance of {\kmccs} if and only if $G'$ has $k(n-1)+$ $k\choose 2$ edge disjoint paths in $\mathcal{P}$.
\end{lemma}

\begin{proof}
For the first direction, let $(G,\{V_1,....,V_k\})$ is a yes instance of {\kmccs}, and that $v_{i,f(i)}$ is the vertex selected from the set $V_i$ in the solution. From each gadget $H_i$, we select paths $P_{{x_{i,l},V_{i,\sigma(l)}}}$ where $l\in[n-1]$, and let $\sigma : [n-1] \to [n]\setminus f(i) $ be any bijection, i.e. one long path corresponding to every $V_{i,j}$ except $V_{i,f(i)}$. This way we are selecting $n-1$ edge disjoint long paths from each gadget which amounts to $k(n-1)$ long paths selected from all the gadgets combined. In every $H_i$, all the edges incident on $v_{i,f(i),l}$ where $l\in[n]$ are free, that is these edges do not belong to any selected long path. For every pair of vertices $v_{i,f(i)}, v_{j,f(j)}$ in the solution of {\kmccs}, w.l.o.g assuming $i<j$, we select short path $P_{v_{i,f(i),j},v_{j,f(j),i}}$, this way we have selected $k\choose 2$ pairwise edge disjoint short paths. 
\paragraph{}

For the other direction, let there be $k(n-1)+$ $k\choose 2$ pairwise edge disjoint paths in $\cal P$. Since there are $k\choose2$ inter gadget edges $E(C)$ and all short paths contain one inter gadget edge, there can be at most $k\choose 2$ short paths in the solution. Thus, there are at least $k(n-1)$ long paths in the solution, recall that for every gadget $H_i$, every long path starts with $x_{i,l}$ where $l\in [n-1]$, further every $x_{i,l}$ is a degree $1$ vertex, thus at most $n-1$ long paths are selected from every gadget. Further, to ensure that at least $k(n-1)$ long paths are selected, there must be $n-1$ long paths selected from every $H_i$, that is one long path corresponding to every set $V_{i,j}$ except one. Let $\sigma: [k] \to [n]$, such that for the gadget $H_i$  no long path $P_{x_{i,l},V_{i,\sigma(i)}}$ is selected for any $l\in [n-1]$. 
\paragraph{}

We argue that $\bigcup_{i=1}^k v_{i,\sigma(i)}$ is a multi colored clique of size $k$ in $G$. Since there are $k(n-1)$ long paths in the solution, there must be $k\choose 2$ short paths. Since no long path corresponding to $V_{i,\sigma(i)}$ for any $i\in [k]$ is selected in the solution, no edge incident on any vertex of $V_{i,\sigma(i)}$ belongs to any selected long path. Further, for each gadget $H_i$, only edges incident on $V_{i,\sigma(i)}$ do not belong to selected long paths. Thus, every short path selected must be a $P_{v_{i,\sigma(i),j},v_{j,\sigma(j),i}}$ for $i<j$, this can only happen if $v_{i,\sigma(i)}$ and $v_{j,\sigma(j)}$ are adjacent in $G$. This finishes the proof.
\end{proof}
\section{Hardness with Respect to Pathwidth $+$ Maximum Degree $+$ Solution Size.}
\begin{figure}[h]
    \centering
    \begin{tikzpicture}[scale=1]

    \node [rectangle,dashed, draw= black!20, minimum height=4.8cm,minimum width= 10cm, label ={-5:$P_{i}$}] (Pi) at (7.3,-3.1) {};

   \node [ label ={0:\tiny{\textcolor{blue}{$P^e_{i,k}$}}}] (Pi) at (12.5,-6) {};
   \node [ label ={0:\tiny{\textcolor{blue}{$P^e_{i,2}$}}}] (Pi) at (12.5,-6.8) {};
   \node [ label ={0:\tiny{\textcolor{blue}{$P^e_{i,1}$}}}] (Pi) at (12.5,-7.6) {};

   \draw[thin,black!60] (2.6,-1.2) to (2.6,-5);
   \draw[thin,black!60] (10.63,-1.2) to (10.63,-5);
    \draw[thin,black!60] (12.03,-1.2) to (12.03,-5);
   
   \draw[thin,black!60] (4.7,-1.2) to (4.7,-5);
   \draw[thin,black!60] (6.4,-1.2) to (6.4,-5);
   \draw[thin,black!60] (8.2,-1.2) to (8.2,-5);

   \draw[thin,black!60] (2.6,-5) to[out=0,in=270] (3.2,-4);
   \draw[thin,black!60] (8.2,-5) to[out=0,in=270] (8.8,-4);
   
   \draw[thin,black!60] (4.7,-5) to[out=0,in=-180] (6.4,-1.2);
   \draw[thin,black!60] (6.4,-5) to[out=0,in=-180] (8.2,-1.2);
   \draw[thin,black!60] (10.3,-1.5) to[out=0,in=180] (10.63,-1.2);
   \draw[thin,black!60] (4.2,-1.7) to[out=0,in=180] (4.7,-1.2);
   \draw[thin,black!30] (10.63,-5) to[out=0,in=180] (12.03,-1.2);
   
   \draw[thin,red!40] (10.63,-0.5) to (10.63,-1.2);
     \draw[thin,red!40] (10.63,-0.5) to (12.03,-1.2);
     \draw[thin,red!40] (2.6,-0.5) to (2.6,-1.2);
     \draw[thin,red!40] (2.6,-0.5) to (3,-0.9);
     \draw[thin,red!40] (4.7,-0.5) to (4.7,-1.2);
     \draw[thin,red!40] (4.7,-0.5) to (6.4,-1.2);
     \draw[thin,red!40] (6.4,-0.5) to (6.4,-1.2);
     \draw[thin,red!40] (6.4,-0.5) to (8.2,-1.2);
     \draw[thin,red!40] (8.2,-0.5) to (8.2,-1.2);
     \draw[thin,red!40] (8.2,-0.5) to (8.5,-0.9);
     \draw[thin,red!40] (9.7,-0.5) to (10.63,-1.2);
     \draw[thin,red!40] (9.7,-0.5) to (9.7,-1.2);
     
     \draw[thin,blue!20] (2.6,-6.0) to (12.03,-6.0);
     \draw[thin,blue!20] (2.6,-6.8) to (12.03,-6.8);
     \draw[thin,blue!20] (2.6,-7.6) to (12.03,-7.6);
     
     \draw[thin,blue!20] (2.6,-1.8) to[out=240,in=110] (2.6,-7.6);
     \draw[thin,blue!20] (2.6,-3.0) to[out=240,in=110] (2.6,-6.8);
     \draw[thin,blue!20] (2.6,-5.0) to[out=240,in=110] (2.6,-6.0);
     
      \draw[thin,blue!20] (4.7,-1.8) to[out=240,in=110] (4.7,-7.6);
     \draw[thin,blue!20] (4.7,-3.0) to[out=240,in=110] (4.7,-6.8);
     \draw[thin,blue!20] (4.7,-5.0) to[out=240,in=110] (4.7,-6.0);
     
      \draw[thin,blue!20] (6.4,-1.8) to[out=240,in=110] (6.4,-7.6);
     \draw[thin,blue!20] (6.4,-3.0) to[out=240,in=110](6.4,-6.8);
     \draw[thin,blue!20] (6.4,-5.0) to[out=240,in=110] (6.4,-6.0);
     
      \draw[thin,blue!20] (8.2,-1.8) to[out=240,in=110] (8.2,-7.8);
     \draw[thin,blue!20] (8.2,-3.0) to[out=240,in=110] (8.2,-6.8);
     \draw[thin,blue!20] (8.2,-5.0) to[out=240,in=110] (8.2,-6);
     
     \draw[thin,blue!20] (10.63,-1.8) to[out=240,in=110](10.63,-7.6);
     \draw[thin,blue!20] (10.63,-3.0) to[out=240,in=110] (10.63,-6.8);
     \draw[thin,blue!20] (10.63,-5.0) to[out=240,in=110](10.63,-6);
     

    \foreach \j in {3}{
        
             \foreach \i in {1}{
              \pgfmathtruncatemacro\result{\j-3}
            \node[shape=circle, fill = blue, scale= 0.3, font=\footnotesize,label={-90:\tiny{\textcolor{blue}{$x_{i,i',k}$}}}] ({vl\j}) at (\j*1.8-0.7,-6.0){};

   }
   }
   \foreach \j in {4,5}{
        
             \foreach \i in {1}{
              \pgfmathtruncatemacro\result{\j-3}
            \node[shape=circle, fill = blue, scale= 0.3, font=\footnotesize,label={-90:\tiny{\textcolor{blue}{$x_{i,i'+\result,k}$}}}] ({vl\j}) at (\j*1.8-0.8,-6.0){};

   }
   }
   
    \foreach \j in {1}{
        
             \foreach \i in {1}{
              \pgfmathtruncatemacro\result{\j-3}
            \node[shape=circle, fill = blue, scale= 0.3, font=\footnotesize,label={-90:\tiny{\textcolor{blue}{$x_{i,1,k}$}}}] ({vl\j}) at (2.6,-6.0){};

   }
   }
   \foreach \j in {7}{
        
             \foreach \i in {1}{
              \pgfmathtruncatemacro\result{\j-3}
            \node[shape=circle, fill = blue, scale= 0.3, font=\footnotesize,label={-90:\tiny{\textcolor{blue}{$x_{i,n,k}$}}}] ({vl\j}) at (-0.8+6.35*1.8,-6.0){};

   }
   }
   \foreach \j in {8}{
        
             \foreach \i in {1}{
              \pgfmathtruncatemacro\result{\j-3}
            \node[shape=circle, fill = blue, scale= 0.3, font=\footnotesize,label={-90:\tiny{\textcolor{blue}{$c_{i,k}$}}}] ({vl\j}) at (-1.2+7.35*1.8,-6.0){};

   }
   }

   \foreach \i in {1,2,3}{
   
            \node[shape=circle, fill = blue, scale= 0.15, font=\footnotesize] ({vl1\i}) at (3.0+\i*0.3,-6.0){};
            \node[shape=circle, fill = blue, scale= 0.15, font=\footnotesize] ({vl1\i}) at (9.0+\i*0.3,-6.0){};

   }

   
   \foreach \j in {3}{
        
             \foreach \i in {1}{
              \pgfmathtruncatemacro\result{\j-3}
            \node[shape=circle, fill = blue, scale= 0.3, font=\footnotesize,label={-90:\tiny{\textcolor{blue}{$x_{i,i',2}$}}}] ({vl\j}) at (\j*1.8-0.7,-6.8){};

   }
   }
   \foreach \j in {4,5}{
        
             \foreach \i in {1}{
              \pgfmathtruncatemacro\result{\j-3}
            \node[shape=circle, fill = blue, scale= 0.3, font=\footnotesize,label={-90:\tiny{\textcolor{blue}{$x_{i,i'+\result,2}$}}}] ({vl\j}) at (\j*1.8-0.8,-6.8){};

   }
   }
   
    \foreach \j in {1}{
        
             \foreach \i in {1}{
              \pgfmathtruncatemacro\result{\j-3}
            \node[shape=circle, fill = blue, scale= 0.3, font=\footnotesize,label={-90:\tiny{\textcolor{blue}{$x_{i,1,2}$}}}] ({vl\j}) at (2.6,-6.8){};

   }
   }
   \foreach \j in {7}{
        
             \foreach \i in {1}{
              \pgfmathtruncatemacro\result{\j-3}
            \node[shape=circle, fill = blue, scale= 0.3, font=\footnotesize,label={-90:\tiny{\textcolor{blue}{$x_{i,n,2}$}}}] ({vl\j}) at (-0.8+6.35*1.8,-6.8){};

   }
   }
   \foreach \j in {8}{
        
             \foreach \i in {1}{
              \pgfmathtruncatemacro\result{\j-3}
            \node[shape=circle, fill = blue, scale= 0.3, font=\footnotesize,label={-90:\tiny{\textcolor{blue}{$c_{i,2}$}}}] ({vl\j}) at (-1.2+7.35*1.8,-6.8){};

   }
   }

   \foreach \i in {1,2,3}{
   
            \node[shape=circle, fill = blue, scale= 0.15, font=\footnotesize] ({vl1\i}) at (3.0+\i*0.3,-6.8){};
            \node[shape=circle, fill = blue, scale= 0.15, font=\footnotesize] ({vl1\i}) at (9.0+\i*0.3,-6.8){};

   }

   
    \foreach \j in {3}{
        
             \foreach \i in {1}{
              \pgfmathtruncatemacro\result{\j-3}
            \node[shape=circle, fill = blue, scale= 0.3, font=\footnotesize,label={-90:\tiny{\textcolor{blue}{$x_{i,i',1}$}}}] ({vl\j}) at (\j*1.8-0.7,-7.6){};

   }
   }
   
   \foreach \j in {4,5}{
        
             \foreach \i in {1}{
              \pgfmathtruncatemacro\result{\j-3}
            \node[shape=circle, fill = blue, scale= 0.3, font=\footnotesize,label={-90:\tiny{\textcolor{blue}{$x_{i,i'+\result,1}$}}}] ({vl\j}) at (\j*1.8-0.8,-7.6){};

   }
   }
   
    \foreach \j in {1}{
        
             \foreach \i in {1}{
              \pgfmathtruncatemacro\result{\j-3}
            \node[shape=circle, fill = blue, scale= 0.3, font=\footnotesize,label={-90:\tiny{\textcolor{blue}{$x_{i,1,1}$}}}] ({vl\j}) at (2.6,-7.6){};

   }
   }
   \foreach \j in {7}{
        
             \foreach \i in {1}{
              \pgfmathtruncatemacro\result{\j-3}
            \node[shape=circle, fill = blue, scale= 0.3, font=\footnotesize,label={-90:\tiny{\textcolor{blue}{$x_{i,n,1}$}}}] ({vl\j}) at (-0.8+6.35*1.8,-7.6){};

   }
   }
   \foreach \j in {8}{
        
             \foreach \i in {1}{
              \pgfmathtruncatemacro\result{\j-3}
            \node[shape=circle, fill = blue, scale= 0.3, font=\footnotesize,label={-90:\tiny{\textcolor{blue}{$c_{i,1}$}}}] ({vl\j}) at (-1.2+7.35*1.8,-7.6){};

   }
   }

   \foreach \i in {1,2,3}{
   
            \node[shape=circle, fill = blue, scale= 0.15, font=\footnotesize] ({vl1\i}) at (3.0+\i*0.3,-7.6){};
            \node[shape=circle, fill = blue, scale= 0.15, font=\footnotesize] ({vl1\i}) at (9.0+\i*0.3,-7.6){};

   }

    \foreach \j in {1}{
        
             \foreach \i in {1}{
              \pgfmathtruncatemacro\result{\j-3}
            \node[shape=circle, fill = red, scale= 0.3, font=\footnotesize,label={0:\tiny{\textcolor{red}{$w_{i,1}$}}}] ({w\j}) at (\j*2.6,-0.5){};

   }
   }

   \foreach \j in {4,5}{
        
             \foreach \i in {1}{
              \pgfmathtruncatemacro\result{\j-3}
            \node[shape=circle, fill = red, scale= 0.3, font=\footnotesize,label={0:\tiny{\textcolor{red}{$w_{i,i'+\result}$}}}] ({w\j}) at (\j*1.8-0.8,-0.5){};

   }
   }
   
   \foreach \j in {3}{
        
             \foreach \i in {1}{
              \pgfmathtruncatemacro\result{\j-3}
            \node[shape=circle, fill = red, scale= 0.3, font=\footnotesize,label={0:\tiny{\textcolor{red}{$w_{i,i'}$}}}] ({w\j}) at (\j*1.8-0.7,-0.5){};

   }
   }
  
   \foreach \j in {7}{
        
             \foreach \i in {1}{
              \pgfmathtruncatemacro\result{\j-3}
            \node[shape=circle, fill = red, scale= 0.3, font=\footnotesize,label={90:\tiny{\textcolor{red}{$w_{i,n-1}$}}}] ({w\j}) at (9.7,-0.5){};

   }
   }
   \foreach \j in {8}{
        
             \foreach \i in {1}{
              \pgfmathtruncatemacro\result{\j-3}
            \node[shape=circle, fill = red, scale= 0.3, font=\footnotesize,label={90:\tiny{\textcolor{red}{$w_{i,n}$}}}] ({w\j}) at (10.63,-0.5){};

   }
   }

   \foreach \i in {1,2,3}{
   
            \node[shape=circle, draw, fill = black, scale= 0.1, font=\footnotesize] ({v1\i}) at (3.0+\i*0.3,-3.3){};
            \node[shape=circle, draw, fill = black, scale= 0.1, font=\footnotesize] ({v1\i}) at (8.9+\i*0.3,-3.3){};

   }
        
             \foreach \i in {1,2}{
              \pgfmathtruncatemacro\result{\i+1}
            \node[shape=circle, draw, fill = black, scale= 0.3, font=\footnotesize,label={0:\tiny{$v_{i,1,\i}$}}] ({v\i1}) at (0.8+1*1.8,-\i*1.2){};

   }
    \foreach \i in {1,2}{
            \node[shape=circle, draw, fill = black, scale= 0.3, font=\footnotesize,label={0:\tiny{$u_{i,1,\i}$}}] ({u\i1}) at (0.8+1*1.8,-0.6-\i*1.2){};
   }

    \foreach \i in {1,2}{
              \pgfmathtruncatemacro\result{\i+1}
            \node[shape=circle, draw, fill = black, scale= 0.3, font=\footnotesize,label={0:\tiny{$v_{i,n,\i}$}}] ({v\i7}) at (-0.8+6.35*1.8,-\i*1.2){};

   }
    \foreach \i in {1,2}{
              \pgfmathtruncatemacro\result{\i+1}
            \node[shape=circle, draw, fill = black, scale= 0.3, font=\footnotesize,label={0:\tiny{$v_{i,n+1,\i}$}}] ({v\i7}) at (-1.2+7.35*1.8,-\i*1.2){};

   }
    \foreach \i in {1,2}{
            \node[shape=circle, draw, fill = black, scale= 0.3, font=\footnotesize,label={0:\tiny{$u_{i,n,\i}$}}] ({u\i7}) at (-0.8+6.35*1.8,-0.6-\i*1.2){};
   }
   
    \foreach \i in {1,2}{
            \node[shape=circle, draw, fill = black, scale= 0.3, font=\footnotesize,label={0:\tiny{$u_{i,n+1,\i}$}}] ({u\i7}) at (-1.2+7.35*1.8,-0.6-\i*1.2){};
   }

     \foreach \i in {6}{
            \node[shape=circle, draw, fill = black, scale= 0.3, font=\footnotesize,label={-20:\tiny{$v_{i,1,k}$}}] ({v\i1}) at (0.8+1*1.8,-0.8-\i*0.6){};
    }
     \foreach \i in {6}{
            \node[shape=circle, draw, fill = black, scale= 0.3, font=\footnotesize,label={-20:\tiny{$u_{i,1,k}$}}] ({v\i1}) at (0.8+1*1.8,-1.4-\i*0.6){};
   }
    \foreach \i in {6}{
            \node[shape=circle, draw, fill = black, scale= 0.3, font=\footnotesize,label={0:\tiny{$v_{i,n,k}$}}] ({v\i7}) at (-0.8+6.35*1.8,-0.8-\i*0.6){};
    }
     \foreach \i in {6}{
            \node[shape=circle, draw, fill = black, scale= 0.3, font=\footnotesize,label={0:\tiny{$u_{i,n,k}$}}] ({u\i7}) at (-0.8+6.35*1.8,-1.4-\i*0.6){};
   }
   
   \foreach \i in {6}{
            \node[shape=circle, draw, fill = black, scale= 0.3, font=\footnotesize,label={0:\tiny{$v_{i,n+1,k}$}}] ({v\i7}) at (-1.2+7.35*1.8,-0.8-\i*0.6){};
    }
     \foreach \i in {6}{
            \node[shape=circle, draw, fill = black, scale= 0.3, font=\footnotesize,label={0:\tiny{$u_{i,n+1,k}$}}] ({u\i7}) at (-1.2+7.35*1.8,-1.4-\i*0.6){};
   }

   \foreach \j in {3}{
        
             \foreach \i in {1}{
              \pgfmathtruncatemacro\result{\i+1}
            \node[shape=circle, draw, fill = black, scale= 0.3, font=\footnotesize,label={0:\tiny{$v_{i,i',\i}$}}] ({v\i\j}) at (\j*1.8-0.7,-\i*1.2){};

   }
    \foreach \i in {2}{
              \pgfmathtruncatemacro\result{\i+1}
            \node[shape=circle, draw, fill = black, scale= 0.3, font=\footnotesize] ({v\i\j}) at (\j*1.8-0.7,-\i*1.2){};

   }
    \foreach \i in {1}{
            \node[shape=circle, draw, fill = black, scale= 0.3, font=\footnotesize] ({u\i\j}) at (\j*1.8-0.7,-0.6-\i*1.2){};
   }
   \foreach \i in {2}{
            \node[shape=circle, draw, fill = black, scale= 0.3, font=\footnotesize] ({u\i\j}) at (\j*1.8-0.7,-0.6-\i*1.2){};
   }
     \foreach \i in {6}{
            \node[shape=circle, draw, fill = black, scale= 0.3, font=\footnotesize] ({v\i\j}) at (\j*1.8-0.7,-0.8-\i*0.6){};
    }
     \foreach \i in {6}{
            \node[shape=circle, draw, fill = black, scale= 0.3, font=\footnotesize,label={270:\tiny{$u_{i,i',k}$}}] ({u\i\j}) at (\j*1.8-0.7,-1.4-\i*0.6){};
   }
   }
   
   \foreach \j in {4,5}{
        
             \foreach \i in {1}{
              \pgfmathtruncatemacro\result{\j-3}
            \node[shape=circle, draw, fill = black, scale= 0.3, font=\footnotesize,label={0:\tiny{$v_{i,i'+\result,\i}$}}] ({v\i\j}) at (\j*1.8-0.8,-\i*1.2){};

            }    
            \foreach \i in {1}{
                    \pgfmathtruncatemacro\result{\j-3}
                    \node[shape=circle, draw, fill = black, scale= 0.3, font=\footnotesize] ({u\i\j}) at (\j*1.8-0.8,-0.6-\i*1.2){};
            }
             \foreach \i in {2}{
              \pgfmathtruncatemacro\result{\j-3}
            \node[shape=circle, draw, fill = black, scale= 0.3, font=\footnotesize] ({v\i\j}) at (\j*1.8-0.8,-\i*1.2){};

            }    
            \foreach \i in {2}{
                    \pgfmathtruncatemacro\result{\j-3}
                    \node[shape=circle, draw, fill = black, scale= 0.3, font=\footnotesize] ({u\i\j}) at (\j*1.8-0.8,-0.6-\i*1.2){};
            }

                \foreach \i in {6}{
                        \pgfmathtruncatemacro\result{\j-3}
                            \node[shape=circle, draw, fill = black, scale= 0.3, font=\footnotesize] ({v\i\j}) at (\j*1.8-0.8,-0.8-\i*0.6){};
                }
            \foreach \i in {6}{
                    \pgfmathtruncatemacro\result{\j-3}
                     \node[shape=circle, draw, fill = black, scale= 0.3, font=\footnotesize,label={270:\tiny{$u_{i,i'+\result,k}$}}] ({u\i\j}) at (\j*1.8-0.8,-1.4-\i*0.6){};
            }

   }

     \foreach \j in {4,5}{
        
             \foreach \i in {3,4,5}{
            \node[shape=circle, draw, fill = black, scale= 0.1, font=\footnotesize] ({v\i\j}) at (\j*1.8-0.8,-2.5-\i*0.3){};

   }
   \foreach \i in {3,4,5}{
            \node[shape=circle, draw, fill = black, scale= 0.1, font=\footnotesize] ({v\i\j}) at (3*1.8-0.7,-2.5-\i*0.3){};

   }
   }
   \foreach \j in {1}{
        
             \foreach \i in {3,4,5}{
            \node[shape=circle, draw, fill = black, scale= 0.1, font=\footnotesize] ({v\i\j}) at (0.8+\j*1.8,-2.5-\i*0.3){};

   }
   }
   \foreach \j in {7}{
        
             \foreach \i in {3,4,5}{
            \node[shape=circle, draw, fill = black, scale= 0.1, font=\footnotesize] ({v\i\j}) at (-0.8+6.35*1.8,-2.5-\i*0.3){};

   }
   }
    \foreach \j in {8}{
        
             \foreach \i in {3,4,5}{
            \node[shape=circle, draw, fill = black, scale= 0.1, font=\footnotesize] ({v\i\j}) at (-1.2+7.35*1.8,-2.5-\i*0.3){};

   }
   }

\end{tikzpicture}
\caption{An example of path $P_i$, edge verification paths $P^e_{i,1}$ $P^e_{i,2}$, and $P^e_{i,k}$, also the edges between vertices of vertex selection paths and edge verification paths. }\label{fig:PWDK}
\end{figure}
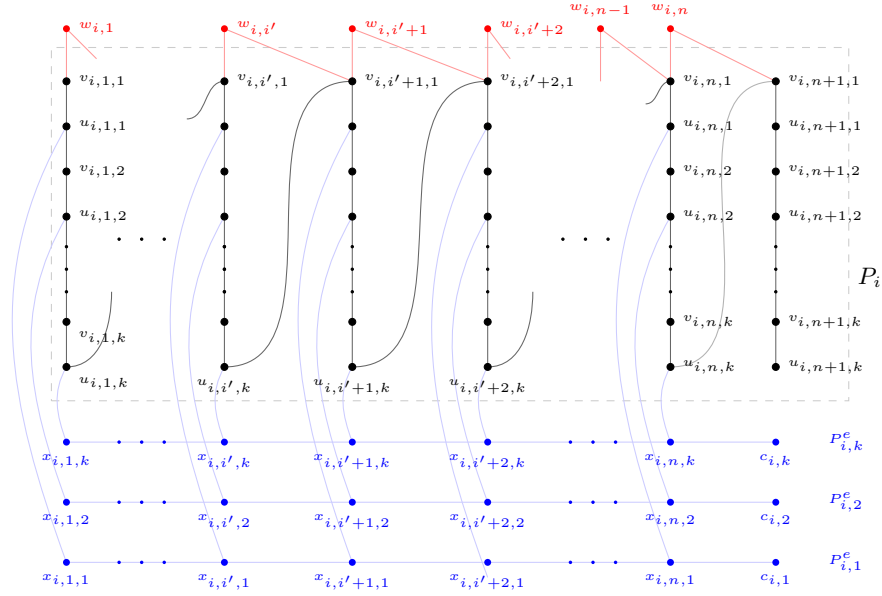

  We give a parameterized reduction from {\kmccs} to {\psp}.
 Let $G=(V,E)$ and $\{V_1,...,V_k\}$ be the input for {\kmccs}, for the  simplicity of notations we assume that every set has $n$ vertices and the vertices of set $V_i$ are labeled $v_{i,1}$ to $v_{i,n}$. We will construct an equivalent instance $(G'=(V',E'),\mathcal{P},k')$ of {\psp} (see Figure  \ref{fig:PWDK} for overview), the construction of $G'$ is as follows.
 
 \begin{itemize}
     \item For every $V_i$, we construct a gadget (subgraph of $G'$) which includes a vertex selection path $P_i$, a vertex set $W_i$, and $k$ edge verification paths $P^e_{i,l}$ where $l\in[k]$ as follows.
 \begin{itemize}
     \item  Corresponding to $V_i$, we start with creating $n+1$ paths of $2k$ vertices each, one path for every vertex $v_{i,i'}\in V_i$, and an additional path. For every $i'\in[n+1]$, the $i'$ path is $(v_{i,i',1},u_{i,i',1}, v_{i,i',2},u_{i,i',2}...., v_{i,i',k},u_{i,i',k})$. That is, it is a sequence of $v_{i,i',l},u_{i,i',l}$ for $l=1$ to $l=k$. We now combine these $n+1$ paths into one path $P_i$ by adding an edge between $u_{i,i',k}$ and $v_{i,i'+1,1}$ for every $i'\in [n]$. 
     
     \item We create $n$ vertices $w_{i,1}$ to $w_{i,n}$ and call the set of these vertices $W_i$. For every $i'\in[n]$, we connect $w_{i,i'}$ to $v_{i,i',1}$ and  $v_{i,i'+1,1}$ .

     \item We create $k$ edge verification paths $P^e_{i,1}$ to  $P^e_{i,k}$ with $n+1$ vertices each. The path $P^e_{i,j}$ is $(x_{i,1,j},x_{i,2,j},...,x_{i,n,j},c_{i,j})$, that is, its a sequence of vertices $x_{i,1,j}$ to $x_{i,n,j}$ and a vertex $c_{i,j}$ at the end.
     
     \item For every $j\in[k]$, $i'\in[n]$, we connect $u_{i,i',j}$ to $x_{i,i',j}$. These edges connects the vertices of vertex selection path $P_i$ to edge verification paths $P^e_{i,j}$.
\end{itemize}
     
     \item After constructing above mentioned gadgets for every vertex set in $\{V_1,..,V_k\}$, for $1\leq i<j \leq k$, we connect $c_{i,j}$ to $c_{j,i}$. We call these edges the inter gadget edges. Observe that there are $k \choose 2$ inter gadget edges.
\end{itemize}
\paragraph{}

    The above completes the construction of $G'$. We now construct collection $\mathcal{P}$ of size $nk+ |E|$ as follows.

    \begin{itemize}
        \item For every $i\in [k]$, from the subgraph of $G'$ induced by $V(P_i)\cup W_i$, we will add $n$ paths in the collection $\cal P$ as follows.
        \begin{itemize}
        \item Add a path $l_{i,\bar  i'}=(P_i(v_{i,1,1},v_{i,i',1}),w_{i,i'},P_i(v_{i,i'+1,1},u_{i,n+1,k}))$ for every $i'\in[n]$, where $P_i(v_{i,1,1},v_{i,i',1})$ is the path from vertex $v_{i,1,1}$ to $v_{i,i',1}$ in $P_i$ (a unique path since $P_i$ is a path). Intuitively, for every $i'\in[n]$ the $l_{i,\bar {i'}}$ contains all the edges of $P_i$ except the edges which belong to subpath $P_i(v_{i,i',1},u_{i,i',k})$. We call these paths the long paths. 
         \end{itemize}
         
        \item For every edge $v_{i,i'}v_{j,j'}\in E$ where $i<j$, we add the path $s_{i,i',j,j'}$= \\$(v_{i,i',j},u_{i,i',j}, P^e_{i,j}(x_{i,i',j},c_{i,j}),P^e_{j,i}(c_{j,i},x_{j,j',i}),u_{j,j',i},v_{j,j',i})$ in $\calp$, where\\ $P^e_{i,j}(x_{i,i',j},c_{i,j})$ is the path from $x_{i,i',j}$ to $c_{i,j}$ in $P^e_{i,j}$ (a unique path, since $P^e_{i,j}$ is a path). We note that every $s_{i,i',j,j'}$ contains exactly one inter gadget edge $c_{i,j},c_{j,i}$. This finishes the construction of $\cal P$.
        
    \end{itemize}
    \paragraph{}

    We set $k'=k+{k\choose2}$. Observe that the construction of $(G',\mathcal{P},k')$ takes time $(|V|)^{O(1)}$. We now claim the bounds on pathwidth and maximum degree of $G'$.

    \begin{lemma}\label{lemma:Pathwidthbound}
    Pathwidth of $G'$ is $O(k^2)$ and maximum degree of $G'$ is $4$.
    \end{lemma}

\begin{proof}
      Recalling the construction of $G'$, every $u_{i,i',j}$ has at most three neighbors, as they are connected to one vertex of edge verification path $P^e_{i,j}$. Every $x_{i,i',j}$ has at most three neighbors. Every $c_{i,j}$ has at most two neighbors. Every $w_{i,i'}$ has at most two neighbors. Every $v_{i,i',j}$ has at most four neighbors, as it can have at most two neighbors from path $P_i$ and at most two neighbors from $W_i$, and these are the vertices with highest degree. Thus maximum degree of $G'$ is at most $4$.
      \paragraph{}

      To bound the pathwidth, we will construct a path decomposition with pathwidth $O(k^2)$ for $G'$. We first decompose the gadgets corresponding to a fixed $i$, which is a sub-graph of $G'$ induced on the vertices $V(P_i)$, $W_i$, and $V(P^e_{i,i'})$. Since these gadgets are similar, we can use the same construction for every $i$.
      \paragraph{}

      For the simplicity of notations, we label $2k(n+1)-1$ edges of $P_{i}$ as $e_{i,1}$ to $e_{i,2k(n+1)-1}$ starting from $v_{i,1,1}u_{i,1,1}$ till $v_{i,n+1,k}u_{i,n+1,k}$. We construct the path decomposition as follows.
      
      \begin{enumerate}
          \item Create a sequence of bags $\beta_i =(B_{i,1},B_{i,2},...,B_{i,2k(n+1)-1})$;
          \item For every $j\in [2k(n+1)-1]$, add $V(e_{i,j})$ in $B_{i,j}$, that is we add every edge of $P_{i}$ sequentially in a sequence of bags;
          \item For every $P^e_{i,j}$ where $j\in[k]$, add $x_{i,i',j}$ in $B_{i,q}$ if $B_{i,q}$ contains at least one vertex from $\{v_{i,i',1},v_{i,i',2},..,v_{i,i',k}\}\cup \{u_{i,i',1},u_{i,i',2},..,u_{i,i',k}\}$;
          \item For every $i'\in[n]$, add $w_{i,i'}$ in $B_{i,q}$ if $B_{i,q}$ contains at least one vertex from\\ $\{v_{i,i',1},v_{i,i',2},...,v_{i,i',k}\}\cup \{u_{i,i',1},u_{i,i',2},...,u_{i,i',k}\}$.
      \end{enumerate}
\paragraph{}

After constructing $\beta_i$. Let $\beta= (\beta_1,\beta_2,...,\beta_k)$. We now add vertices $c_{i,j}$ for every $i,j\in[k]$ into every bag of every $\beta$. That is every $c_{i,j}$ belongs to every bag. We claim that $\beta$ represents the path decomposition of $G'$.

\paragraph{}

We first bound the maximum number of vertices in a bag. In the step 2, we added two vertices in each bag which are either $\{v_{i,i',l},u_{i,i',l}\}$ or $\{u_{i,i',k},v_{i,i'+1,1}\}$ for $i'\in[n+1]$, this implies that in step 3 we are adding at most $2k$ vertices in every bag $B_q$. Similar to step 3, addition of vertices $w_{i,i'}$ in step 4 depends on the vertices added in step $2$ and this step adds at most $2$ vertices in every bag. Finally, we add $k^2$ vertices $c_{i,j}$ in every bag. Thus, every bag contains at most $O(k^2)$ vertices, that is the pathwidth is $O(k^2)$. 
\paragraph{}

For the correctness of the decomposition, it is straightforward to verify that every vertex of $G'$ belong to at least one of the bag. Step 2 of construction ensures that for every edge $e_j\in P_{i}$, there is a bag $B_{i,j}$ which contain $V(e_j)$. All the edges between a vertex of an edge verification path $P^e_{i,j}$ and a vertex of a vertex selection path $P_i$ are of the form $u_{i,i',j}x_{i,i',j}$, and step 3 ensures that for every such edge there is a bag which contain them. The edges of edge verification paths $P^e_{i,j}$ of the form $x_{i,i',j}x_{i,i'+1,j}$ will be added (step 3) to the bag which contains edge $u_{i,i',k}v_{i,i'+1,1}$ (its an edge in $P_i$ and existence of such a bag is ensured by step 2). Similarly, the edges of form $w_{i,i'}v_{i,i',1}$ will be added (step 4) to the bag which contains vertex $v_{i,i',1}$, and the edges of the form $w_{i,i'}v_{i,i'+1,1}$ will be added by step 4 to the bag containing $u_{i,i',k},v_{i,i'+1,1}$ (existence of such bag assured by step 2). All the remaining edges incidents on a vertex $c_{i,j}$ and these vertices are added to every bag.
\paragraph{}

We now verify that if a vertex $v\in G'$ belongs to bags $B_p$ and $B_q$ where $p<q$, then for every $p<p'<q$, $v\in B_{p'}$. For the vertices of $P_{i}$, this property can be verified by the fact that step 2 adds all the edges of  $P_{i}$ sequentially and that violation of this property will contradict that $P_i$ is a path. The step 3 adds vertices $x_{i,i',j}$ in the bags which contain at least one vertex from $\{v_{i,i',1},v_{i,i',2},...,v_{i,i',k}\}\cup \{u_{i,i',1},u_{i,i',2},...,u_{i,i',k}\}$, since these vertices forms a subpath of $P_i$, and thus, they are in consecutive bags (step 1) and violation would contradict that $P_i$ is a path, finally vertices $c_{i,j}$ are added to every bag. The step 4 adds vertices $w_{i,i'}$ in the bags which contain at least one vertex from $\{v_{i,i',1},v_{i,i',2},...,v_{i,i',k}\}\cup \{u_{i,i',1},u_{i,i',2},...,u_{i,i',k}\}$, since these vertices forms a subpath of $P_i$, and thus, they are in consecutive bags (step 1) and violation would contradict that $P_i$ is a path. This concludes the correctness of path decomposition $\beta$ and finishes the proof.
\end{proof}

\paragraph{}

The following concludes the correctness of reduction and proof of Theorem \ref{theorem:Pathwidth-hardness}.
    
     \begin{lemma}\label{lemma:pathwidthcorrectness}
    $G=(V,E)$ with partition $V_1$ to $V_k$ is a yes instance of {\kmccs} if and only if $\mathcal{P}$ has $k+{k\choose2}$ pairwise edge disjoint paths.
    \end{lemma}

  \begin{proof}
     
     For the first direction, let $(G,E,\{V_1,....,V_k\})$ be a yes instance of {\kmccs}, and let $v_{i,f(i)}$ be the vertex selected from the set $V_i$ in the solution. We chose $k$ long paths $l_{i,\bar {f(i)}}$ from $\cal P$ for every $i\in[k]$. Since for every $1\leq i< j\leq k$, $v_{i,f(i)},v_{j,f(j)}$ is an edge in $G$, we chose $k\choose 2$ short paths $s_{i,f(i),j,f(j)}$ from $\cal P$. Recall that selecting $l_{i,\bar {f(i)}}$ will utilize all the edges of $P_i$ except the edges of the subpath $P_i(v_{i,i',1},v_{i,i',k})$, which now can be utilized by short paths $s_{i,f(i),j,f(j)}$. Further, a direct check can verify that all these selected paths are pairwise edge disjoint.
     
     \paragraph{}

        For the other direction let there be a solution $S$ of size $k+$ $k\choose2$ for {\psp}. Recall that all the short paths  goes through one of the inter gadget edge $c_{i,j}c_{j,i}$, and there are $k\choose2$ such edges. Thus, we can conclude that there are at most $k\choose2$ short paths in the solution. Thus, at least $k$ long paths are selected in the solution. Since for every $i\in[k]$, all the long paths $l_{i,\bar {i'}}$ where $i'\in[n]$ are pair wise edge intersecting (contain edge $v_{i,n+1,1},u_{i,n+1,1}$). Thus, at most one long path can be selected for every $i$. Since, at least $k$ long paths need to be selected, we conclude that exactly one long path is selected corresponding to every $i\in[k]$.
        Let for every $i$, $l_{i,\bar {f(i)}}$ be the long path selected in $S$. Let $C= \{v_{i,f(i)}| l_{i,\bar {f(i)}} \in S\}$, that is all the vertices of $G$ whose corresponding long paths are selected in the solution. We argue that $C$ is a multi colored clique of size $k$, assume to the contrary that $v_{i,f(i)},v_{j,f(j)}\in C$ where $i<j$ are not neighbors in $G$. 
        \paragraph{}

        Since there are $k\choose 2$ short paths selected and $k\choose 2$ inter gadget edges $c_{i,j},c_{j,i}$ available, every inter gadget edge $c_{i,j},c_{j,i}$ must belong to one selected short path. Recalling the construction of short paths, a short path that contain the edge $c_{i,j},c_{j,i}$ (where $i<j$),
        starts with an edge $v_{i,i',j}u_{i,i',j}$ of  path $P_i$ and ends at an edge $u_{j,j',i}v_{j,j',i}$ of  path $P_j$ where $i',j'\in[n]$. Due to the selection of long paths $l_{i,\bar{f(i)}}$ in $S$, all the edges of $P_i$ except the edges which belong to subpath $P_i(v_{i,f(i),1},u_{i,f(i),k})$ are utilised.
         Thus, all the $k\choose 2$ short paths selected in $S$ must start with an edge $v_{i,f(i),j}u_{i,f(i),j}$ and ends with an edge $u_{j,f(j),i}v_{j,f(j),i}$, formally these short paths should be $s_{i,f(i),j,f(j)}$ for every $1\leq i<j\leq k$, this can only happen if $v_{i,f(i)}v_{j,f(j)}$ is an edge in $G$, contradicting the assumption.
         This finishes the proof.
\end{proof}
\section{Parameterized Algorithms}

In section \ref{section:FVN}, we present an FPT algorithm for PSP parameterized by feedback vertex number $\Gamma$ plus maximum degree $\Delta$ of $G$, and prove Theorem \ref{theorem:FVD_maxdegree}. 
In section \ref{sect:FENalg}, we present a $4$-approximation algorithm for finding maximum path set packing which runs in time FPT parameterized by feedback edge number $\lambda$ and prove Theorem  \ref{theorem:4-apxFEN}. The approach we use to prove Theorem \ref{theorem:FVD_maxdegree} and Theorem \ref{theorem:4-apxFEN} is a non trivial adaptation of the approach given in \cite{Jansen17}.

\paragraph{}

In section \ref{sect:FENPre}, we discuss some preliminaries which we will need for the algorithms presented in sections \ref{section:FVN} and \ref{sect:FENalg}. In section \ref{sec:3parameters} we give the proof of Theorem \ref{theorem:tw_maxdegree}.
\paragraph{}

 We note that the algorithms we present can be modified slightly to solve the maximization version of {\psp}; in this version, we search for the path set packing of maximum size.

\subsection{Preliminaries: Defining Structures and Nice Solutions} \label{sect:FENPre}
Let $G=(V,E)$ be the input undirected graph. If $G$ is not connected, then we connect $G$ as follows. Let there be $m$ components in $G$, and let $\{C_1,C_2,\dots C_m\}$ be the set of all the components of $G$.  We arbitrarily pick a vertex from every  component of $G$, let $v_i$ be the vertex picked from the component $C_i$. For every $i\in [m-1]$, we add an edge $(v_i,v_{i+1})$. Let $((v_1,v_2),(v_2,v_3),\dots, (v_{m-1},v_{m}))$ be the order in which these edges are added to $G$. Observe that each time we are adding an edge between two distinct components of $G$. Thus, addition of these edges does not introduce any new cycle in $G$, and hence does not increase the feedback vertex number or the feedback edge number of $G$. Further, $\cal P$ is not changed, and maximum degree of $G$ is increased by at most two. From now on we therefore assume that the input graph $G$ is connected.
\paragraph{}

We now modify the input graph $G$ as follows. We add vertex set $\{z_1,z_2,z_3\}$ to $G$, and arbitrarily pick a vertex $v\in G$ and connect vertices in $\{z_1,z_2,z_3,v\}$ to each other, thus creating a clique on four vertices in $G$. This modification increases the feedback vertex number,  feedback edge number, and maximum degree of $G$ by only a constant, and it is safe for our purposes. In the remaining part of this section, we define some structures (adapted from \cite{Jansen17}) that we will need. 

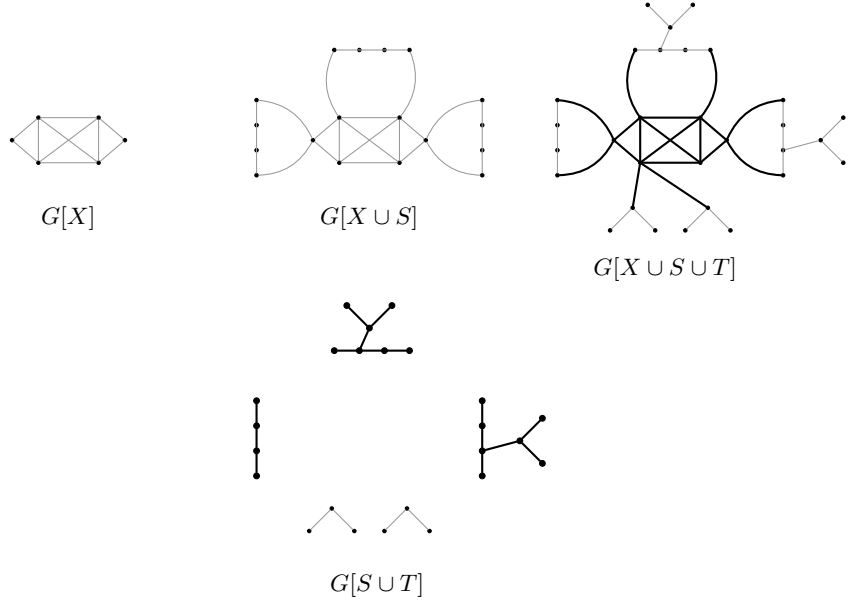
\begin{figure}[h]

    \centering
    \begin{tikzpicture}


 \node [ellipse, minimum height=1.5cm,minimum width= 2cm, label ={-90:\footnotesize{\text{$G[X]$}}}] (z) at (2*1,-2*1) {};
   
   \node[shape=circle, draw, fill = black, scale= 0.15, font=\footnotesize] ({11}) at (2*1+0.4,-2*1+0.3){};
    \node[shape=circle, draw, fill = black, scale= 0.15, font=\footnotesize] ({12}) at (2*1-0.4,-2*1+0.3){};
    \node[shape=circle, draw, fill = black, scale= 0.15, font=\footnotesize] ({13}) at (2*1+0.4,-2*1-0.3){};
    \node[shape=circle, draw, fill = black, scale= 0.15, font=\footnotesize] ({14}) at (2*1-0.4,-2*1-0.3){};
     \node[shape=circle, draw, fill = black, scale= 0.15, font=\footnotesize] ({15}) at (2*1-0.75,-2*1-0){};
     \node[shape=circle, draw, fill = black, scale= 0.15, font=\footnotesize] ({16}) at (2*1+0.75,-2*1-0){};
   
   \draw[thin,black!40] ({11}) to ({12});
   \draw[thin,black!40] ({11}) to ({13});
   \draw[thin,black!40] ({11}) to ({14});
   \draw[thin,black!40] ({11}) to ({16});
   \draw[thin,black!40] ({13}) to ({16});
   \draw[thin,black!40] ({12}) to ({13});
   \draw[thin,black!40] ({12}) to ({14});
   \draw[thin,black!40] ({13}) to ({14});
   \draw[thin,black!40] ({12}) to ({15});
   \draw[thin,black!40] ({14}) to ({15});
   
    \node [ellipse, minimum height=1.5cm,minimum width= 2cm, label ={-90:\footnotesize{\text{$G[X\cup S]$}}}] (z) at (2*3,-2*1) {};
   
   \node[shape=circle, draw, fill = black, scale= 0.15, font=\footnotesize] ({21}) at (2*3+0.4,-2*1+0.3){};
    \node[shape=circle, draw, fill = black, scale= 0.15, font=\footnotesize] ({22}) at (2*3-0.4,-2*1+0.3){};
    \node[shape=circle, draw, fill = black, scale= 0.15, font=\footnotesize] ({23}) at (2*3+0.4,-2*1-0.3){};
    \node[shape=circle, draw, fill = black, scale= 0.15, font=\footnotesize] ({24}) at (2*3-0.4,-2*1-0.3){};
     \node[shape=circle, draw, fill = black, scale= 0.15, font=\footnotesize] ({25}) at (2*3-0.75,-2*1-0){};
     \node[shape=circle, draw, fill = black, scale= 0.15, font=\footnotesize] ({26}) at (2*3+0.75,-2*1-0){};
   
   \draw[thin,black!40] ({21}) to ({22});
   \draw[thin,black!40] ({21}) to ({23});
   \draw[thin,black!40] ({21}) to ({24});
   \draw[thin,black!40] ({21}) to ({26});
   \draw[thin,black!40] ({23}) to ({26});
   \draw[thin,black!40] ({22}) to ({23});
   \draw[thin,black!40] ({22}) to ({24});
   \draw[thin,black!40] ({23}) to ({24});
   \draw[thin,black!40] ({22}) to ({25});
   \draw[thin,black!40] ({24}) to ({25});
   
   \foreach \i in {1,2,3,4}{
     
            \node[shape=circle, draw, fill = black, scale= 0.15, font=\footnotesize,label={}] ({21\i}) at (2*3-0.8+\i/3,-2*1+1.2){};
   }
   \foreach \i in {1,2,3,4}{
     
            \node[shape=circle, draw, fill = black, scale= 0.15, font=\footnotesize,label={}] ({22\i}) at (2*3+1.5,-2*1-0.8+\i/3){};
   }
   
    \foreach \i in {1,2,3,4}{
     
            \node[shape=circle, draw, fill = black, scale= 0.15, font=\footnotesize,label={}] ({23\i}) at (2*3-1.5,-2*1-0.8+\i/3){};
   }
   \draw[thin,black!40] ({211}) to ({214});
   \draw[thin,black!40] ({221}) to ({224});
   \draw[thin,black!40] ({231}) to ({234});
   
   \draw[thin,black!40,bend right =30] ({211}) to ({22});
   \draw[thin,black!40,bend left =30] ({214}) to ({21});
    \draw[thin,black!40,bend right =30] ({224}) to ({26});
   \draw[thin,black!40,bend left =30] ({221}) to ({26});
    \draw[thin,black!40,bend right =30] ({231}) to ({25});
   \draw[thin,black!40,bend left =30] ({234}) to ({25});

  
    \node [ellipse, minimum height=1.5cm,minimum width= 2cm, label ={-90:\footnotesize{\text{}}}] (z) at (2*5,-2*1) {};
   
   \node[shape=circle, draw, fill = black, scale= 0.15, font=\footnotesize] ({31}) at (2*5+0.4,-2*1+0.3){};
    \node[shape=circle, draw, fill = black, scale= 0.15, font=\footnotesize] ({32}) at (2*5-0.4,-2*1+0.3){};
    \node[shape=circle, draw, fill = black, scale= 0.15, font=\footnotesize] ({33}) at (2*5+0.4,-2*1-0.3){};
    \node[shape=circle, draw, fill = black, scale= 0.15, font=\footnotesize] ({34}) at (2*5-0.4,-2*1-0.3){};
     \node[shape=circle, draw, fill = black, scale= 0.15, font=\footnotesize] ({35}) at (2*5-0.75,-2*1-0){};
     \node[shape=circle, draw, fill = black, scale= 0.15, font=\footnotesize] ({36}) at (2*5+0.75,-2*1-0){};
   
   \draw[thick,black] ({31}) to ({32});
   \draw[thick,black] ({31}) to ({33});
   \draw[thick,black] ({31}) to ({34});
   \draw[thick,black] ({31}) to ({36});
   \draw[thick,black] ({33}) to ({36});
   \draw[thick,black] ({32}) to ({33});
   \draw[thick,black] ({32}) to ({34});
   \draw[thick,black] ({33}) to ({34});
   \draw[thick,black] ({32}) to ({35});
   \draw[thick,black] ({34}) to ({35});

   \foreach \i in {1,2,3,4}{
     
            \node[shape=circle, draw, fill = black, scale= 0.15, font=\footnotesize,label={}] ({31\i}) at (2*5-0.8+\i/3,-2*1+1.2){};
   }
   \foreach \i in {1,2,3,4}{
     
            \node[shape=circle, draw, fill = black, scale= 0.15, font=\footnotesize,label={}] ({32\i}) at (2*5+1.5,-2*1-0.8+\i/3){};
   }
   
    \foreach \i in {1,2,3,4}{
     
            \node[shape=circle, draw, fill = black, scale= 0.15, font=\footnotesize,label={}] ({33\i}) at (2*5-1.5,-2*1-0.8+\i/3){};
   }
   \draw[thin,black!40] ({311}) to ({314});
   \draw[thin,black!40] ({321}) to ({324});
   \draw[thin,black!40] ({331}) to ({334});
   
   \draw[thick,black,bend right =30] ({311}) to ({32});
   \draw[thick,black,bend left =30] ({314}) to ({31});
    \draw[thick,black,bend right =30] ({324}) to ({36});
   \draw[thick,black,bend left =30] ({321}) to ({36});
    \draw[thick,black,bend right =30] ({331}) to ({35});
   \draw[thick,black,bend left =30] ({334}) to ({35});

   \node[shape=circle, draw, fill = black, scale= 0.15, font=\footnotesize,label={}] ({351}) at (2*5-0,-2*1+1.5){};
   \node[shape=circle, draw, fill = black, scale= 0.15, font=\footnotesize,label={}] ({352}) at (2*5-0.3,-2*1+1.8){};
   \node[shape=circle, draw, fill = black, scale= 0.15, font=\footnotesize,label={}] ({353}) at (2*5+0.3,-2*1+1.8){};
   
   \draw[thin,black!40] ({351}) to ({352});
   \draw[thin,black!40] ({351}) to ({353});
  \draw[thin,black!40] ({351}) to ({312});
   
   \node[shape=circle, draw, fill = black, scale= 0.15, font=\footnotesize,label={}] ({361}) at (2*5-0.5,-2*1-0.9){};
   \node[shape=circle, draw, fill = black, scale= 0.15, font=\footnotesize,label={}] ({362}) at (2*5-0.8,-2*1-1.2){};
   \node[shape=circle, draw, fill = black, scale= 0.15, font=\footnotesize,label={}] ({363}) at (2*5-0.2,-2*1-1.2){};
   
      \draw[thin,black!40] ({361}) to ({362});
   \draw[thin,black!40] ({361}) to ({363});
  \draw[thick,black] ({361}) to ({34});

  \node[shape=circle, draw, fill = black, scale= 0.15, font=\footnotesize,label={}] ({381}) at (2.2*5-0.5,-2*1-0.9){};
   \node[shape=circle, draw, fill = black, scale= 0.15, font=\footnotesize,label={}] ({382}) at (2.2*5-0.8,-2*1-1.2){};
   \node[shape=circle, draw, fill = black, scale= 0.15, font=\footnotesize,label={}] ({383}) at (2.2*5-0.2,-2*1-1.2){};
   
      \draw[thin,black!40] ({381}) to ({382});
   \draw[thin,black!40] ({381}) to ({383});
  \draw[thick,black] ({381}) to ({34});
   
    \node[shape=circle, draw, fill = black, scale= 0.15, font=\footnotesize,label={}] ({371}) at (2*5+2,-2*1+0){};
   \node[shape=circle, draw, fill = black, scale= 0.15, font=\footnotesize,label={}] ({372}) at (2*5+2.3,-2*1-0.3){};
   \node[shape=circle, draw, fill = black, scale= 0.15, font=\footnotesize,label={}] ({373}) at (2*5+2.3,-2*1+0.3){};
   
   \draw[thin,black!40] ({371}) to ({372});
   \draw[thin,black!40] ({371}) to ({373});
  \draw[thin,black!40] ({371}) to ({322});
  
  \node [ label ={-20:\footnotesize{\text{$G[X\cup S \cup T]$}}}] (z) at (1.75*5,-2.6*1-0.8) {};


   \foreach \i in {1,2,3,4}{
     
            \node[shape=circle, draw, fill = black, scale= 0.25, font=\footnotesize,label={}] ({41\i}) at (1.2*5-0.8+\i/3,-6*1+1.2){};
   }
   \foreach \i in {1,2,3,4}{
     
            \node[shape=circle, draw, fill = black, scale= 0.25, font=\footnotesize,label={}] ({42\i}) at (1.2*5+1.5,-6*1-0.8+\i/3){};
   }
   
    \foreach \i in {1,2,3,4}{
     
            \node[shape=circle, draw, fill = black, scale= 0.25, font=\footnotesize,label={}] ({43\i}) at (1.2*5-1.5,-6*1-0.8+\i/3){};
   }
   \draw[thick,black] ({411}) to ({414});
   \draw[thick,black] ({421}) to ({424});
   \draw[thick,black] ({431}) to ({434});

\node[shape=circle, draw, fill = black, scale= 0.25, font=\footnotesize,label={}] ({451}) at (1.2*5-0,-6*1+1.5){};
   \node[shape=circle, draw, fill = black, scale= 0.25, font=\footnotesize,label={}] ({452}) at (1.2*5-0.3,-6*1+1.8){};
   \node[shape=circle, draw, fill = black, scale= 0.25, font=\footnotesize,label={}] ({453}) at (1.2*5+0.3,-6*1+1.8){};
   
   \draw[thick,black] ({451}) to ({452});
   \draw[thick,black] ({451}) to ({453});
  \draw[thick,black] ({451}) to ({412});
   
   \node[shape=circle, draw, fill = black, scale= 0.15, font=\footnotesize,label={}] ({461}) at (1.2*5-0.5,-6*1-0.9){};
   \node[shape=circle, draw, fill = black, scale= 0.15, font=\footnotesize,label={}] ({462}) at (1.2*5-0.8,-6*1-1.2){};
   \node[shape=circle, draw, fill = black, scale= 0.15, font=\footnotesize,label={}] ({463}) at (1.2*5-0.2,-6*1-1.2){};
   
   \draw[thin,black!40] ({461}) to ({462});
   \draw[thin,black!40] ({461}) to ({463});

   \node[shape=circle, draw, fill = black, scale= 0.15, font=\footnotesize,label={}] ({481}) at (1.4*5-0.5,-6*1-0.9){};
   \node[shape=circle, draw, fill = black, scale= 0.15, font=\footnotesize,label={}] ({482}) at (1.4*5-0.8,-6*1-1.2){};
   \node[shape=circle, draw, fill = black, scale= 0.15, font=\footnotesize,label={}] ({483}) at (1.4*5-0.2,-6*1-1.2){};
   
   \draw[thin,black!40] ({481}) to ({482});
   \draw[thin,black!40] ({481}) to ({483});

    \node[shape=circle, draw, fill = black, scale= 0.25, font=\footnotesize,label={}] ({471}) at (1.2*5+2,-6*1+0){};
   \node[shape=circle, draw, fill = black, scale= 0.25, font=\footnotesize,label={}] ({472}) at (1.2*5+2.3,-6*1-0.3){};
   \node[shape=circle, draw, fill = black, scale= 0.25, font=\footnotesize,label={}] ({473}) at (1.2*5+2.3,-6*1+0.3){};
   
   \draw[thick,black] ({471}) to ({472});
   \draw[thick,black] ({471}) to ({473});
  \draw[thick,black] ({471}) to ({422});
  
  \node [ label ={-20:\footnotesize{\text{$G[S\cup T]$}}}] (z) at (1.05*5,-6.8*1-0.8) {};

\end{tikzpicture}
\caption{Example induced graphs of $G$. The darkened edges in $G[X\cup S\cup T]$ represents the edge set ${\cal A}=E(G[X]) \cup E(X,S\cup T)$. The darkened components in $G[S\cup T]$ forms the set $\cal D$, and the lighter components forms $\cal T$. }\label{fig:structures}
\end{figure}

\begin{definition}
Given a graph $G$, we define the vertex sets $T,S$ and $X$ by the following process (refer Figure  \ref{fig:structures}).
\begin{itemize}
    \item Initialize $T$ as an empty set and $G'=(V',E')$ as a copy of $G$.
    \item While there is a vertex $v$ in $G'$ with degree $d
    _{G'}(v)=1$, we set $T=T\cup\{v\}$ and $G'=G'- \{v\}$.
    \item  $X$ is the set of all the vertices with degree at least three in $G[V\setminus T]$, and $S$ is the set $V\setminus (T\cup X)$.
\end{itemize}
\end{definition}

\paragraph{}

Observe that $G[T]$ is a forest as it is a $1$-degenerate graph. Every component of $G[T]$ is connected to $G[X\cup S]$ by a single edge in $G$ (Figure \ref{fig:structures}). $G[X\cup S]$ is connected and its minimum degree is at least two . Vertex set $X$ is nonempty as it contains at least vertices $z_1,z_2,z_3$, and $v$ (as they form a clique). Every vertex in $S$ has degree two in $G[X\cup S]$. Further, we observe that $G[S]$ is a union of paths, otherwise if $G[S]$ contains a cycle $C$, then all the vertices of $C$ will have degree $2$ in $G[S]$ itself, and they will be disconnected from the vertices of $X$ in $G[X\cup S]$, implying that $G[X\cup S]$ is disconnected, which is not the case.

\begin{observation}\label{obs:atmost1}
Every component of $G[S\cup T]$ contains at most one component of $G[S]$.
\end{observation}

\begin{proof}
Every component of $G[T]$ is connected to $G[X\cup S]$ by a single edge. Thus, no component of $G[T]$ is connected to more than one component of $G[S]$ in $G[S\cup T]$. Hence, components of $G[S]$ are not merged in  $G[S\cup T]$.
\end{proof}

\paragraph{}

Let $\cal D$ be the set of all the components of $G[S\cup T]$ which contain a component of $G[S]$ (refer Fig \ref{fig:structures}). Further, let $\cal T $ be the set of all other components in $G[S\cup T]$, that is the set of all the components which do not contain vertices of $S$. Further, we observe that every component $C\in {\cal D}\cup \cal T$ is a tree. 

\begin{definition}
    For every $C\in{\cal D}\cup \cal T$, we define $\textsc{ext\_e}(C)= E(V(C),X)$, and call every edge of $\textsc{ext\_e}(C)$ an external edge of $C$.
\end{definition}

\begin{observation}\label{obs:twoedges}
Every $D\in \cal D$ has two external edges, that is $|\textsc{ext\_e}(D)|=2$. And every $C\in \cal T$ has one external edge, that is $|\textsc{ext\_e}(C)|=1$.
\end{observation}

\begin{proof}
Every component of $G[T]$ is connected to $G[X\cup S]$ by a single edge in $G$. Since $\cal T$ is a subset of all the components of $G[T]$, and none of the component in $\cal T$ is connected to a vertex in $S$ (by definition of $\cal T$), there must be a single edge between every component in $\cal T$ and vertex set $X$ in $G$. 
\paragraph{}

By definition, every vertex in $S$ has degree two in $G[X\cup S]$, and every component of $G[S]$ is a path. Thus, if a component of $G[S]$ is a path of two or more vertices, then each endpoint of it must be connected to a vertex in $X$ in $G$, and if the component is a single vertex, then that vertex must be connected to two vertices in $X$ in $G$. Thus, there are two edges between every component in $G[S]$ and the vertex set $X$ in $G$.
Further, consider a component $C\in \cal D$. By definition and by Observation \ref{obs:atmost1}, $C$ contains exactly one component $C'$ of $G[S]$. Hence, there are at least two edges between $V(C)$ and $X$ in $G$. Further, if $C$ also contains components of $G[T]$, then each of these components is connected to a vertex in $V(C')$, which is a subset of $S$; thus, they cannot have an edge with a vertex in $X$ in $G$. Thus, there are exactly two edges between $V(C)$ and $X$ in $G$.
\end{proof}

\paragraph{}

To bound the size of $X$ and $\cal D$, we recall the following from \cite{Jansen17}.

\begin{proposition}[\cite{Jansen17}]
Let $G$ be a connected graph of minimum degree at least two with cyclomatic number (feedback edge number) $\lambda$. Let $X$ be the set of all the vertices of degree at least three in $G$, then $|X| \leq 2 \lambda-2$, and if $X \neq \emptyset$, then the number of connected components of $G[V\setminus X]$ is at most $\lambda+|X|-1$.
\end{proposition}

\paragraph{}

We get the following corollary.
\begin{corollary}\label{cor:sizeX}
The size of vertex set $|X|= O(\lambda) = O(\Gamma \cdot \Delta)$. The size of component set $|{\cal D}| = O(\lambda) = O(\Gamma \cdot \Delta)$. Where $\Gamma$ is feedback vertex number of $G$.
\end{corollary}

\paragraph{}

In our algorithms, we will construct sub-problems for subgraphs of $G$, for which we need to define internal paths of subgraphs. 

\begin{definition}
For a subgraph $H\subseteq G$, a path $p$ is an internal path of $H$ if $E(p)\subseteq E(H)$, that is all the edges of $p$ belongs to $H$. Further, given a set of paths $P$ in $G$, we define $\textsc{int}(H,{P})= \{ p\mid \ p\in P \land (E(p) \subseteq E(H))\}$, that is all the paths in $P$ which are internal paths of $H$.
\end{definition}

\begin{observation}\label{obs:atleast_one_ext}
    Let $C\in {\cal D}\cup \cal T$, and let $p$ be a non internal path of $C$. If $E(p)\cap \textsc{ext\_e}(C) = \emptyset$, then $E(p)\cap E(C)=\emptyset$.
\end{observation}
\begin{proof}
It follows from the fact that if $p$ is a non internal path of $C$, then $p$ contains a vertex $v$ outside $C$, and every path in $G$ from the vertex $v$ to a vertex in $C$ contains an external edge of $C$.
\end{proof}

\begin{lemma}\label{lemma:paththreeormore1}
   Every $p\in {\cal P}$ can contain exactly one external edge of at most two distinct components of ${\cal D}\cup {\cal T}$.
\end{lemma}
\begin{proof}
    Assume to the contrary that a path $p\in{\cal P}$ contains exactly one external edge of three or more components of ${\cal D} \cup {\cal T}$, let $C_1,C_2,C_3\in {\cal D} \cup {\cal T}$ be any three of them, then $p$ has at least one vertex each from $C_1,C_2$, and $C_3$, let them be $v_1,v_2,$ and $v_3$ respectively. Without loss of generality, let $(v_1,v_2,v_3)$ be a subsequence of $p$, that is $v_1$ comes before $v_2$, and $v_2$ comes before $v_3$ in $p$. In this case, the subpath of $p$ from $v_1$ to $v_2$ contains an external edge $e_1$ of $C_2$ (as $v_1\not\in V(C_2))$, also the subpath of $p$ from $v_2$ to $v_3$ contains an external edge $e_2$ of $C_2$ (as $v_3\not\in V(C_2))$. Since $p$ is a simple path, $e_1\neq e_2$, contradicting the assumption that only one external edge of $C_2$ belongs to $p$.
\end{proof}

\begin{definition}
For an instance $(G,{\cal P},k)$ of {\psp}, for a subgraph $H\subseteq G$, we define $\textsc{opt}(H)$ to be the maximum size of a path set packing in $\textsc{int}(H,{\cal P})$.
\end{definition}
\paragraph{}

For a graph $G$ in context, we simply use vertex sets $T,S,X$ and component set $\cal D$ and $\cal T$ as defined above, unless otherwise stated.

\begin{definition}[nice path set packing]\label{definition:nicesolution}
Let $(G,{\cal P}, k)$ be an instance of {\psp}. We say, a path set packing $M\subseteq \cal P$ is a nice path set packing if the following holds.
\begin{itemize}
    \item for every component $D\in \cal D$, ${\textsc{opt}}(D)\geq |{\textsc{int}}(D,M)|\geq \textsc{opt}(D)-1$;
    \item for every component $T\in \cal T$, $\textsc{opt}(T)= |\textsc{int}(T,M)|$.
\end{itemize}
\end{definition}
\paragraph{}

The following lemma will help bound the number of guesses that we have to make in our algorithm, and it is motivated from the ideas used in \cite{Jansen17}.
\begin{lemma} \label{lemma:maxnicepath}
    Let $(G,{\cal P}, k)$ be an instance of {\psp}, and $M\subseteq \cal P$ be a path set packing of maximum size, then there exist a nice path set packing $M'\subseteq \cal P$ such that $|M'|=|M|$.
\end{lemma}
\begin{proof}
For every path set packing $M\subseteq \cal P$, it follows trivially that for every $C\in {\cal D}\cup {\cal T}$, ${\textsc{opt}}(C)\geq |{\textsc{int}}(C,{M})|$.
For the other direction of the inequality, given a maximum path set packing $M\subseteq \cal P$, we modify $M$ as follows.
\begin{itemize}
    \item While there exists a $C \in \cal D$, such that $ | {\textsc{int}}(C,{M}) |\leq {\textsc{opt}}(C)-2$, (OR) there exists $C \in \cal T$, such that $ | {\textsc{int}}(C,{M}) |\leq {\textsc{opt}}(C)-1$, we do the following:
    \begin{enumerate}
        \item Let $P=\{p\mid p\in M\land (E(p)\cap E(C)\neq \emptyset)\}$.
        \item Set $M= M\setminus P$.
        \item Let $P^*$ be a path set packing of size $\textsc{opt}(C)$ in ${\textsc{int}}(C,{\cal P})$, Set $M= M\cup P^*$.
    \end{enumerate}
\end{itemize}
\paragraph{}

  We first argue that in step $2$, we are removing at most $|P|\leq \textsc{opt}(C)$ paths from $M$. Every path in $P \setminus {\textsc{int}}(C,M)$ contains an edge of $C$ and is not an internal path of $C$, hence it must contain an external edge of $C$. Since $P$ is a path set packing, $P \setminus {\textsc{int}}(C,M)$ can contain at most as many paths as the number of external edges of $C$. Thus, if $C\in \cal D$, then  $P \setminus {\textsc{int}}(C,{M})$ contains at most two paths, and 
  $|P \cap {\textsc{int}}(C,M)|= |{\textsc{int}}(C,M) |\leq {\textsc{opt}}(C)-2$, hence $|P|\leq {\textsc{opt}}(C)$. Similar arguments work when $C\in {\cal T}$.
  In step $3$, we are adding $ \textsc{opt}(C)$ internal paths of $C$ to $M$. Thus, the size of $M$ does not decrease after each iteration of these modifications. Further, $M$ remains a path set packing, as removal of $P$ form $M$ in step $2$ ensures that $M$ no more contains a path which contains an edge of $C$, and hence adding $P^*$ to $M$ in step $3$ is safe. 
  Further, the procedure converges, since in the steps $2$ and $3$, removal and addition of paths corresponding to a component $C$ do not change $|\textsc{int}(C',M)|$ for every $C'\in {\cal D}\cap {\cal T}$ where $C'\neq C$. Further, after step $3$, $C$ no longer satisfy the loop condition, hence after each iteration there is one less component for which loop condition satisfies.
\end{proof}
\paragraph{}

 We recall that problem of finding maximum independent set in a graph can be solved in polynomial time on EPT graphs \cite{TARJAN1985221}, and it is equivalent to finding a maximum size path set packing in a collection of simple paths in a tree.

\begin{proposition}[\cite{TARJAN1985221,XZ2018}]\label{proposition:psp_tree}
Given a tree $T$ and a collection of simple paths $\cal P$ in $T$, the maximum size of a path set packing in $\cal P$ can be computed in time $(|V|+|\calp|)^{O(1)}$. 
\end{proposition}
\begin{corollary}
For every $C\in {\cal D}\cup {\cal T}$, ${\textsc{opt}}(C)$ can be computed in time $(|V|+|\calp|)^{O(1)}$.
\end{corollary}
\paragraph{}

We also note that by the results of \cite{TARJAN1985221,XZ2018}, for every $C\in {\cal D}\cup {\cal T}$, a path set packing of maximum size (i.e. of size ${\textsc{opt}}(C)$) in $\textsc{int}(C,\cal P)$ can be computed in time $(|V|+|\calp|)^{O(1)}$. 

\begin{corollary}\label{cor:optforest}
Given a subgraph $F$ of $G$ such that $F$ is a forest, ${\textsc{opt}}(F)$ can be computed in time $(|V|+|\calp|)^{O(1)}$. Further, a path set packing of maximum size (i.e. of size ${\textsc{opt}}(F)$) in $\textsc{int}(F,\cal P)$ can be computed in time $(|V|+|\calp|)^{O(1)}$.
\end{corollary}
\begin{proof}
    We can compute ${\textsc{opt}}(c)$ in time $(|V|+|\calp|)^{O(1)}$ for every component $c$ of $F$. Since the components of $F$ are pairwise edge disjoint, ${\textsc{opt}}(F))$ is equal to  $\sum_{c \text{ is a component of }F}{\textsc{opt}}(c)$. Similarly let $P_C$ be a path set packing of size ${\textsc{opt}}(C)$ in $\textsc{int}(C,\cal P)$, then $P_F=\bigcup_{c \text{ is a component of }F}P_C$ is a path set packing of size ${\textsc{opt}}(F))$ in $\textsc{int}(F,\cal P)$.
\end{proof}

\subsection{FPT Algorithm Parameterized by Feedback Vertex Number + Maximum Degree} \label{section:FVN}

In this section, we present an algorithm which solves {\psp} in time $({\lambda \cdot \Delta})^{O(\lambda \cdot \Delta)}\cdot (|V|+|\calp|)^{O(1)}$, which is equivalent to $({\Gamma \cdot \Delta^2})^{O(\Gamma \cdot \Delta^2)}\cdot (|V|+|\calp|)^{O(1)}$, where $\Gamma$ is feedback vertex number of $G$. Given an instance $(G,{\cal P},k)$ of {\psp},
we first define a set of edges. 

\begin{definition}
  We define ${\cal A}= E(G[X]) \cup E(X,S\cup T)$, that is, the set of all edges in $G$ with at least one endpoint in $X$ (refer Figure  \ref{fig:structures}).
\end{definition}

\paragraph{}

We recall that $|X|= O(\lambda)$ (Corollary \ref{cor:sizeX}), and hence $|{\cal A}|= O(\lambda \cdot \Delta)= O(\Gamma \cdot \Delta^2)$. We recall that the vertex set $X$ contains the vertices $z_1,z_2$, and $z_3$ that we added separately, and thus $\cal A$ contains edges between these vertices. Further, no path in $\cal P$ contains any of these edges. This is an essential property that we will use later in our algorithm.
\paragraph{}

We now search for a path set packing of maximum size in $\cal P$, which is also a nice path set packing, and its existence is ensured by Lemma \ref{lemma:maxnicepath}.  Recall that ${\cal D}\cup {\cal T}$ is the set of all the components in $G[S\cup T]$, and since $E(G)\setminus {\cal A}$ is the edge set $E(G[S\cup T])$, every path in the solution which does not intersect with $\cal A$, will be an internal path of a component in ${\cal D} \cup {\cal T}$. We first guess the number of internal paths that every $D\in \cal D$ will have in the solution. Formally, let $f_d: {\cal D} \to \{0,1\}$ be our guess that for every $D\in \cal D$, ${\textsc{opt}}(D)- f_d(D)$ internal paths of $D$ are in the solution. Further, for every $T\in \cal T$ we know that ${\textsc{opt}}(T)$ internal paths of $T$ will be in the solution. 
\paragraph{}

It is now left to guess the number of paths in the solution which intersects with ${\cal A}$. To this end, let $f_e = \{A_1,A_2,\dots, A_{|f_e|}\}$ be a partition of $\cal A$. Intuitively, this is our second guess, where we are guessing that if a path $p$ in the solution intersects with $\cal A$, then there exists an $A\in f_e$ such that $A=E(p)\cap {\cal A}$. Further, we are also guessing that there exists an $A'\in f_e$ such that $A'$ is the set of all the edges of $\cal A$ which do not belong to any path in the solution. Thus, we are guessing that there will be $|f_e|-1$ paths in the solution which intersects with ${\cal A}$. Further, we define  $\textsc{size}(f_d,f_e)$ as follows.

\begin{equation} \label{equation:sizefdfe}
    \textsc{size}(f_d,f_e) = (|f_e|-1)+ \sum_{T\in {\cal T}} {\textsc{opt}}(T) + \sum_{D\in {\cal D}} ({\textsc{opt}}(D)- f_d(D)).
\end{equation}

\paragraph{}

We now formally define the intuition behind the pair $(f_d,f_e)$ as the following.

\begin{definition} \label{definition:feasiblesolution}
  A path set packing $M\subseteq \cal P$ is a feasible solution for a pair $(f_d,f_e)$ if the following holds.
\begin{enumerate}
    \item Let $P=\{p\mid p\in M \land (E(p)\cap {\cal A} \neq \emptyset)\}$, that is the set of all the paths in $M$ which intersects with ${\cal A}$, then
    \begin{itemize}
        \item $|P|=|f_e|-1$;
        \item for every $p\in P$, there exists an $A\in f_e$ such that $E(p)\cap {\cal A} = A$.
    \end{itemize}
    \item For every $D\in \cal D$, $|{\textsc{int}}(D,M)| = {\textsc{opt}}(D)-f_d(D)$.
    \item For every $T\in \cal T$, $| {\textsc{int}}(T,M)| = {\textsc{opt}}(T)$.
\end{enumerate}
\end{definition}
\paragraph{}

Further, we say a pair $(f_d,f_e)$ has a feasible solution if and only if there exists a path set packing $M\subseteq \cal P$ which is a feasible solution for $(f_d,f_e)$.

\begin{observation} \label{obs:sizefdfe}
    If a path set packing $M \subseteq {\cal P}$ is a feasible solution for the pair $(f_d,f_e)$, then $|M|= \textsc{size}(f_d,f_e)$.
\end{observation}
\begin{proof}
    By definition, $M$ contains $|f_e|-1$ paths which intersect with $\cal A$. Since $E(G)\setminus \cal A$ is the edge set $E(G[S\cup T])$, if a path in $M$ does not intersect with $\cal A$, then it must be an internal path of a component in $ {\cal D}\cup {\cal T}$. In Equation \ref{equation:sizefdfe} we are counting every path of $M$ which is an internal path of a component in ${\cal D}\cup {\cal T}$. Further, a path cannot be an internal path of two distinct components, as components are edge disjoint. Hence, we are not counting any path twice.
\end{proof}

\begin{lemma} \label{lemma:maxfeasible}
    For every nice path set packing $M\subseteq \cal P$, there exists a pair $(f_d,f_e)$ such that $M$ is a feasible solution for the pair $(f_d,f_e)$.
\end{lemma}
\begin{proof}
Let $P$ be the set of all the paths in $M$ which intersects with $\cal A$; we construct $f_e = \{E(p) \cap {\cal A} \mid  p\in P\}$ $\cup \{{\cal A}\setminus E(P)\}$. Here $\{{\cal A}\setminus E(P)\}$ is non-empty, as it contains the edges between the vertices of $\{z_1,z_2,z_3\}$, which we added separately. Further, $E(p) \cap {\cal A}$ is distinct, disjoint, and nonempty for every 
$p\in P$, as $P$ is a path set packing. Thus, $f_e$ is a partition of $\cal A$, and $|P|= |f_e|-1$. For the second condition, since $M$ is a nice path set packing, for every component $D\in \cal D$, ${\textsc{opt}}(D)\geq |{\textsc{int}}(D,M)|\geq \textsc{opt}(D)-1$. Thus, for every $D\in \cal D$, if $|{\textsc{int}}(D,M)|=\textsc{opt}(D)-1$, then we set $f_d(D)= 1$, else we set $f_d(D)= 0$. Property three is trivially satisfied by $M$.
\end{proof}
\paragraph{}

Let $M^*\subseteq \cal P$ be a nice path set packing of maximum size in $\cal P$.  Further, let $\cal F$ be the set of all the possible pairs $(f_d,f_e)$, where $f_d: {\cal D}\to \{0,1\}$, and $f_e$ is a partition of $\cal A$. Then
\begin{equation} \label{eqn:sizemstar}
    |M^*|= \max \{ \textsc{size}(f_d,f_e)\mid (f_d,f_e)\in {\cal F} \land (f_d,f_e) \text{ has a feasible solution} \}.
\end{equation}
\paragraph{}

The correctness of Equation \ref{eqn:sizemstar} is due to the following. Observation \ref{obs:sizefdfe} ensures that if $(f_d,f_e)$ has a feasible solution, then there exists a path set packing  $M\subseteq \cal P$ of size $\textsc{size}(f_d,f_e)$. This combined with the maximality of $M^*$ implies that $|M^*|\geq \textsc{size}(f_d,f_e)$ if $(f_d,f_e)$ has a feasible solution. Further, the equality holds due to the Lemma \ref{lemma:maxfeasible}, which ensures that there exists a pair $(f^*_d,f^*_e)\in \cal F$ such that $M^*$ is a feasible solution of $(f^*_d,f^*_e)$, and hence $|M^*|= \textsc{size}(f^*_d,f^*_e)$.

 There are at most $2^{O(\lambda)}$ distinct $f_d: {\cal D}\to \{0,1\}$ and at most $({\lambda \cdot \Delta})^{O(\lambda \cdot \Delta)}$  distinct partitions $f_e$ of $\cal A$. If we can verify whether a pair $(f_d,f_e)$ has a feasible solution in time $(|V|+|{\cal P}|)^{O(1)}$, then we can use Equation \ref{eqn:sizemstar} to find $|M^*|$, and we can bound the running time of the algorithm to $({\lambda \cdot \Delta})^{O(\lambda \cdot \Delta)}\cdot (|V|+|\calp|)^{O(1)}$. In the remaining part of this section, we will discuss how to verify if a pair $(f_d,f_e)$ has a feasible solution. For this purpose, consider the following definition.

\begin{definition} \label{def:fdfecompatible}
For an $\alpha \in f_e$, a set of paths $P\subseteq \cal P$ is compatible to $(f_d,f_e,\alpha)$ if the following hold.
\begin{enumerate}
    \item $P$ is a path set packing;
    \item for every $p\in P$, there exists an $A\in (f_e\setminus \{\alpha\})$ such that $E(p)\cap {\cal A} =A$;
    \item for every $D\in {\cal D}$, ${\textsc{opt}}(D- E(P))\geq ({\textsc{opt}}(D)-f_d(D))$;
    \item for every $T\in {\cal T}$, ${\textsc{opt}}(T- E(P))= {\textsc{opt}}(T)$.
\end{enumerate}
\end{definition}

\begin{observation} \label{obs:subsetcomp}

Given a set of paths $P\subseteq \calp$ and $\alpha \in f_e$, we can verify if $P$ is compatible to $(f_d,f_e,\alpha)$ in time $(|V|+|\calp|)^{O(1)}$. Further, if $P$ is compatible to $(f_d,f_e,\alpha)$, then every subset $P'\subseteq P$ is compatible to $(f_d,f_e,\alpha)$.
\end{observation}

\begin{proof}
Given a set of paths $P$, it is straight-forward to verify the first two conditions of  being compatible to $(f_d,f_e,\alpha)$. Further, for every $C\in {\cal D}\cup {\cal T}$, $C- E(P)$ is a forest, hence using Corollary $\ref{cor:optforest}$ we can compute ${\textsc{opt}}(C-E(P))$ in time $(|V|+|\calp|)^{O(1)}$. Further, if $P$ is compatible to $(f_d,f_e,\alpha)$, then every subset $P'$ of $P$ also satisfies all four conditions of being compatible to $(f_d,f_e,\alpha)$.
\end{proof}

\begin{lemma}\label{lemma:fdfesolution}
$(f_d,f_e)$ has a feasible solution if and only if there exists a set of paths $P\subseteq \cal P$ such that $|P|=|f_e|-1$ and $P$ is compatible to $(f_d,f_e,\alpha)$ for an $\alpha\in f_e$. 
\end{lemma}

\begin{proof}
For the forward direction, let $M\subseteq \cal P$ be a path set packing which is a feasible solution for $(f_d,f_e)$. Let $P$ be the set of all the paths in $M$ which intersects with $\cal A$. By definition, $|P|=|f_e|-1$, hence there exists an $\alpha \in f_e$ such that for every path $p\in P$, $E(p)\cap {\cal A}\neq  \alpha$. It is now straightforward to verify that $P$ satisfies the first two conditions of being compatible to $(f_d,f_e,\alpha)$. Further, for every $D\in \cal D$, the set  ${\textsc{int}}(D,M)$ is a path set packing containing only internal paths of $D$. Since both $P$ and ${\textsc{int}}(D,M)$ are disjoint subsets of a path set packing  $M$, none of the path in ${\textsc{int}}(D,M)$ contains an edge of $E(P)$. Thus, all the paths in ${\textsc{int}}(D,M)$ are internal to $D-E(P)$, that is  ${\textsc{int}}(D,M)$ = ${\textsc{int}}(D-E(p),M)$, and the existence of ${\textsc{int}}(D,M)$ ensures that there exists a path set packing of size $|{\textsc{int}}(D,M)|$ in the set ${\textsc{int}}(D-E(p),\cal P)$. Thus, $\textsc{opt}(D-E(P))\geq |{\textsc{int}}(D,M)|=(\textsc{opt}(D)-f_d(D))$. Thus, $P$ satisfies the third condition of being compatible to $(f_d,f_e,\alpha)$. For the fourth condition, similar arguments as above work for every $T\in \cal T$.
\paragraph{}

For the other direction, let $P\subseteq \cal P$ be a set of paths such that $|P|=|f_e|-1$ and $P$ is compatible to $(f_d,f_e,\alpha)$. We construct a path set packing $M'\subseteq \cal P$ as follows. 
\begin{itemize}
    \item We add $P$ to $M'$;
    \item for every $D\in {\cal D}$, ${\textsc{opt}}(D-E(P))\geq {\textsc{opt}}(D)-f_d(D)$. Thus, for every $D\in \cal D$, let $P_D$ be a path set packing of size  ${\textsc{opt}}(D-E(P))$ in $\textsc{int}(D-E(P),{\cal P})$, we arbitrarily pick a subset of $P_D$ of size exactly ${\textsc{opt}}(D)-f_d(D)$ and add it to $M'$;
    \item for every $T\in {\cal T}$, let $P_T$ be a path set packing of size  ${\textsc{opt}}(T-E(P))$ in $\textsc{int}(T-E(P),{\cal P})$, we add $P_T$ to $M'$.
\end{itemize}
\paragraph{}

 $M'$ is a path set packing, because by definition $P$ is a path set packing, and in steps two and three of the construction, each time we are adding a path set packing in $M'$, which contains internal paths of a distinct component in ${\cal D} \cup {\cal T}$, and components in  ${\cal D} \cup {\cal T}$ are pairwise edge disjoint. Further, none of the paths added in the steps two and three edge intersects with a path in $P$. 
\paragraph{}

In the construction above, apart from paths in $P$, no other path which intersects with $\cal A$ is added to $M'$. Further, adding $P$ to $M'$ ensures that $M'$ satisfies the first property of being a feasible solution for $(f_d,f_e)$. 
In the second step of the construction, for every $D\in \cal D$, we are adding ${\textsc{opt}}(D)-f_d(D)$ internal paths of $D$ in $M'$, and these are the only internal paths of $D$ added to $M'$. Thus, $M'$ satisfies the second property of being a feasible solution for $(f_d,f_e)$. Similarly, in the third step of the construction, for every $T\in \cal T$, we are adding ${\textsc{opt}}(T)$ internal paths of $T$, and these are the only internal paths of $T$ added to $M'$. Thus, $M'$ satisfies the third property of being a feasible solution for $(f_d,f_e)$.
\end{proof}

\paragraph{}

Due to the above lemma, to verify if a pair $(f_d,f_e)$ has a feasible solution, it will suffice to verify if there exists a set of paths $P\subseteq \cal P$ such that the size of $P$ is $|f_e|-1$ and $P$ is compatible to $(f_d,f_e,\alpha)$ for an $\alpha \in f_e$.
Here, the choice of $\alpha$ is essentially automatic, recall that $\cal A$ contains a non zero number of edges which belong to no path in $\cal P$ (at least the edges between the vertices of $\{z_1,z_2,z_3\}$ which we added separately), and if a set $A\in f_e$ contains any of these edges, then for every path $p\in \cal P$, $E(p)\cap {\cal A} \neq A$. If on a choice of $\alpha$, $A$ remains in $f_e\setminus \{\alpha\}$, then every set of paths $P$, which is compatible to $(f_d,f_e,\alpha)$, can be of size at most $|f_e|-2$. This is because a set of paths $P$ must be a path set packing if it is compatible to $(f_d,f_e,\alpha)$, and if $P$ is a path set packing, then for every $p\in P$, $E(p)\cap \cal A$ is distinct. The second condition in Definition \ref{def:fdfecompatible} requires $E(p)\cap \cal A$ to belong to $f_e\setminus \{\alpha\}$, and we know that for every $p\in P$, $E(P)\cap \cal A\neq A$, hence $|P|\leq |f_e|-2$. Essentially, it is required that all the edges in $\cal A$, which belong to no path in $\cal P$, belong to a same set in $f_e$, and we need to choose that set as $\alpha$. Otherwise, we conclude that $(f_d,f_e)$ has no feasible solution. Thus, we shall henceforth assume that there exists a set $A\in f_e$, which contains all the edges in $\cal A$ which belong to no path in $\cal P$, and we chose $A$ as $\alpha$. Further, if for this chosen $\alpha$, there exists a path set packing $P$ of size $|f_e|-1$, which is compatible to $(f_d,f_e,\alpha)$, we conclude that $(f_d,f_e)$ has a feasible solution, otherwise we conclude that $(f_d,f_e)$ has no feasible solution.
\paragraph{}

In the remaining part of this section, we assume that $(f_d,f_e,\alpha)$ is given, and our task is to verify in time  $(|V|+|{\cal P}|)^{O(1)}$, if there exists a path set packing of size of $|f_e|-1$, which is compatible to $(f_d,f_e,\alpha)$. To this end, consider the following lemma.

\begin{lemma}\label{lemma:threeormore}
If there exists an $A\in f_e$ such that $A$ contains exactly one external edge of three or more components in ${\cal D} \cup {\cal T}$, then for every  $p\in \cal P$, $E(p)\cap {\cal A} \neq A$.
\end{lemma}

\begin{proof}

Let there be an $A\in f_e$ such that $A$ contains exactly one external edge of three or more components in ${\cal D} \cup {\cal T}$. Let $p\in \cal P$ be a path such that $E(p)\cap {\cal A} = A$.
 By definition of $\cal A$, union of external edges of every component in ${\cal D} \cup {\cal T}$ forms a subset of $\cal A$. Thus $E(p)\cap {\cal A} = A$ implies that $p$ contains exactly one external edge of three or more components in ${\cal D} \cup {\cal T}$. By Lemma \ref{lemma:paththreeormore1}, we conclude that this is not possible.
\end{proof}

\paragraph{}

By the above lemma, if an $A\in (f_e \setminus \{\alpha\})$ contains exactly one external edge of three or more components in ${\cal D} \cup {\cal T}$, then every set of paths $P\subseteq \cal P$ compatible to $(f_d,f_e,\alpha)$ can be of size at most $|f_e|-2$. This is because if a set $P$ is a path set packing, then for every $p\in P$, $E(p)\cap \cal A$ is distinct. The second condition in Definition \ref{def:fdfecompatible} requires $E(p)\cap \cal A$ to belong to $f_e\setminus \{\alpha\}$, and we know that for every $p\in P$, $E(P)\cap \cal A\neq A$, hence $|P|\leq |f_e|-2$. Thus, we shall henceforth assume that no $A\in (f_e \setminus \{\alpha\})$ contains exactly one external edge of three or more components in ${\cal D} \cup {\cal T}$.
\paragraph{}

We now construct an auxiliary graph $H$ on vertex set $f_e\setminus \{\alpha\}$. In $H$, two vertices $A$ and $A'$ are adjacent if and only if there exists a $D\in \cal D$ such that both sets $A$ and $A'$ contain an external edge of $D$. We claim that the maximum degree of $H$ is two. For the proof of this claim consider the following arguments. If there is an edge between two vertices $A$ and $A'$ in $H$, then there must be a $D\in \cal D$ such that both sets $A$ and $A'$ contain an external edge of $D$, in this case both $A$ and $A'$ contains exactly one external edge of $D$ as $A$ and $A'$ are disjoint and $D$ has only two external edges. Thus, every edge incident on a vertex $A$ in $H$ implies that set $A$ contains exactly one external of a component in $\cal D$. Further, every $D\in \cal D$ can cause at most one edge in $H$, as $D$ has only two external edges, and none of its external edges belongs to more than one set in $f_e$. Thus, if the degree of a vertex $A$ is three or more in $H$, then it implies that the set $A$ contains exactly one external edge of three or more components in $\cal D$, which is not the case. Hence, the maximum degree of $H$ is  at most two, and $H$ is a disjoint union of paths and cycles.
\paragraph{}

 Let there be $l$ components in $H$. Consider an arbitrary but fixed and distinct labelling from $1$ to $l$ of the components of $H$. For every $i\in[l]$,  we construct a sequence $\pi_i$ of vertices of component $i$ as follows.
\begin{itemize}
    \item If the component $i$ is a path: we set $\pi_i$ to be the path formed by component $i$ that contain all its vertices, the starting point of the path is arbitrarily picked from one of its terminal (degree one vertex) if $i$ has two or more vertices.

    \item If the component $i$ is a cycle: we arbitrarily pick a vertex $A$ in component $i$. Since $i$ is a cycle, there are only two simple paths in $i$ which starts with $A$ and contain all the vertices of $i$. We arbitrarily pick one of these two paths, let $p$ be the picked path, we set $\pi_i=p$.
    
\end{itemize}
\paragraph{}

 Since $f_e\setminus \{\alpha \}$ is the vertex set of $H$, we can also say that, for every $i\in [l]$, $\pi_i$ is a sequence of sets of $f_e\setminus \{ \alpha \}$, and this is what we will refer to as $\pi_i$ for our purposes. Further, let $\Pi =\{\pi_{i}|i\in [l]\}$. We note that the construction of $H$ and $\Pi$ takes time $(|V|+|\calp|)^{O(1)}$.

\begin{definition}
  A sequence $\rho$ of paths in $\cal P$ is a candidate for a sequence $\pi\in \Pi$, if the following holds.
        \begin{itemize}
            \item $|\rho|=|\pi|$, and
            \item for every $i \in [|\pi|]$, $E({\rho[i]) \cap {\cal A}} = \pi[i]$.
        \end{itemize}
\end{definition}

\begin{lemma}\label{lemma:fdfesolution_candidate}
There exists a set of paths $P\subseteq \cal P$ such that $|P|=|f_e|-1$ and $P$ is compatible to $(f_d,f_e,\alpha)$ if and only if for every sequence $\pi \in \Pi$, there exists a candidate $\rho$ such that $set(\rho)$ is compatible to $(f_d,f_e,\alpha)$.
\end{lemma}

\begin{proof}
For the forward direction, let $P\subseteq \cal P$ be a set of paths such that $|P|=|f_e|-1$ and $P$ is compatible to $(f_d,f_e,\alpha)$. Since $P$ is a path set packing, for every $p\in P$, $E(p)\cap \cal A$ is distinct and it belongs to $f_e\setminus \{\alpha\}$, this implies that for every $A\in f_e\setminus \{\alpha\}$, there exists a path $p\in P$ such that $E(p)\cap {\cal A} = A$. For every $\pi \in \Pi$, we construct a candidate $\rho_{\pi}$ as follows. For every $i\in [|\pi|]$, let $p$ be the path in $P$ such that $E(p)\cap {\cal A}= \pi[i]$, we set $\rho_{\pi} [i]= p$. Further, for the constructed $\rho_{\pi}$, $set(\rho_{\pi})$ is compatible to $(f_d,f_e,\alpha)$ because it is a subset of $P$. 
\paragraph{}

For the other direction, for every $\pi \in \Pi$, let $\rho_{\pi}$ be a candidate of $\pi$ such that $set(\rho_{\pi})$ is compatible to $(f_d,f_e,\alpha)$. Let $\Phi= \{\rho_{\pi} \mid \pi \in \Pi\}$. Further, let $P^*= \bigcup_{\rho \in \Phi} set(\rho)$. By construction, every $A\in f_e\setminus \{\alpha\}$ belongs to exactly one sequence in $\Pi$, hence the size of $P^*$ is $|f_e|-1$. We now show that $P^*$ is a path set packing. Consider any two distinct paths $p_i,p_j \in P^*$ and let $E(p_i)\cap {\cal A} = A_i$ and $E(p_j)\cap {\cal A} = A_j$. 
Since $A_i$ and $A_j$ are distinct and disjoint, $p_i$ and $p_j$ do not intersect at an edge that belong to $\cal A$. If there exists a $D\in \cal D$ such that both $p_i$ and $p_j$ intersect at an edge of $D$, then by Observation \ref{obs:atleast_one_ext}, both of them must contain an external edge of $D$ as well. This implies that both $A_i$ and $A_j$ must contain an external edge of $D$ as external edges of $D$ forms a subset of $\cal A$. In this case, $A_i$ and $A_j$ must be neighbors in the auxiliary graph $H$, hence they both belong to a same component of $H$, and to a same sequence $\pi \in \Pi$. Thus, both the paths $p_i$ and $p_j$ belong to the candidate $\rho_{\pi}$. Since $set(\rho_{\pi})$ is compatible to $(f_d,f_e,\alpha)$, both $p_i$ and $p_j$ are edge disjoint, contradicting the assumption that they intersect at an edge of $D$. Similarly, every $T\in \cal T$ has only one external edge and that belongs to $\cal A$. Thus, the external edge of every $T\in \cal T$ belongs to at most one of $A_i$ and $A_j$, hence it belongs to at most one of $p_i$ and $p_j$. This implies that for every $T\in \cal T$, at most one of $p_i$ and $p_j$ can contain an edge of $T$ (By Observation \ref{obs:atleast_one_ext}). Thus, $p_i$ and $p_j$ do not intersect at an edge that belong to a $T\in \cal T$.
\paragraph{}

Further, the second condition for $P^*$ being compatible to $(f_d,f_e,\alpha)$ is trivially satisfied. We prove the third condition in Definition \ref{def:fdfecompatible} for an arbitrary $D \in \cal D$. Let $P_D$ be the set of all the paths in $P^*$ which contain an edge of $D$. Since every path in $P_D$ is not internal to $D$, every path in $P_D$ will have an external edge of $D$ as well (Observation \ref{obs:atleast_one_ext}). For every $p\in P_D$, let $A_p= E(p)\cap \cal A$. By construction of $H$, for every distinct $p_1,p_2\in P_D$, $A_{p_1}$ and $A_{p_2}$ are neighbours in $H$. Hence, the vertices $\{A_p \mid p\in P_D\}$ form a connected component in $H$, and hence they all belong to the same sequence $\pi \in \Pi$, and thus $P_D \subseteq set(\rho_{\pi})$. Thus, $(D-E(P^*))= (D-E(P_D))= (D-E(set(\rho_{\pi}))$. Since it is given that $set(\rho_{\pi})$ is compatible to $(f_d,f_e,\alpha)$, we have ${\textsc{opt}}(D-E(P^*))= {\textsc{opt}}(D-E(set(\rho_{\pi})))\geq ({\textsc{opt}}(D)-f_d(D))$. Similarly, for the fourth condition, for every $T\in \cal T$, at most one path in $P^*$ can contain the external edge of $T$, and only that path can contain its edges. If a $p\in P^*$ contains an edge of $T$, then let $\rho\in \Phi$ be the candidate which contains $p$, then $(T-E(P^*))=(T-E(p))=(T-set(\rho))$. Since it is given that $set(\rho)$ is compatible to $(f_d,f_e,\alpha)$, we have ${\textsc{opt}}(T-E(P^*))= {\textsc{opt}}(T-E(set(\rho)))={\textsc{opt}}(T)$.
\end{proof}

\paragraph{}

Due to the above lemma, to verify if there exists a set of paths $P\subseteq \cal P$ such that $|P|=|f_e|-1$ and $P$ is compatible to $(f_d,f_e,\alpha)$, it suffices to verify if for every sequence $\pi \in \Pi$, there exists a candidate $\rho$ such that $set(\rho)$ is compatible to $(f_d,f_e,\alpha)$.
Since $|\Pi|$ is at most $|V|^2$, it suffices to prove that given a sequence $\pi\in \Pi$, in time $(|V|+|\calp|)^{O(1)}$ we can verify if $\pi$ has a candidate $\rho$ such that $set(\rho)$ is compatible to $(f_d,f_e,\alpha)$. To this end, consider the following lemma.

\begin{lemma}\label{lemma:adjacent}
For every type sequence $\pi\in \Pi$, let $\rho$ be a candidate of $\pi$, then $set(\rho)$ is compatible to  $(f_d,f_e,\alpha)$ if and only if
\begin{itemize}
    \item For every $j\in [|\rho|]$, $\{\rho[j],\rho[(j+1)\mod|\rho|]\}$ is compatible to $(f_d,f_e,\alpha)$.
\end{itemize}
\end{lemma}

\begin{proof}
For the forward direction, let $\rho$ be a candidate of $\pi$ such that $set(\rho)$ is compatible to $(f_d,f_e,\alpha)$, then every subset of $set(\rho)$ is compatible to $(f_d,f_e,\alpha)$.
\paragraph{}

For the other direction, let $\rho$ be a candidate of $\pi$ and for every $j\in[|\rho|]$, $\{\rho[j],\rho[(j+1)\mod|\rho|]\}$ is compatible to $(f_d,f_e,\alpha)$. We now show that $\rho$ is compatible to $(f_d,f_e,\alpha)$.  We assume that $|\rho|\geq 3$, otherwise the case is trivial to verify. We first show that $set(\rho)$ is a path set packing. Assume to the contrary that there exists two paths $\rho[i]$ and $\rho[j]$ which intersects at an edge $e\in E(G)$ where $i\neq j$. The edge $e$ cannot belong to $\cal A$, since $\pi[i]$ and $\pi[j]$ are distinct and disjoint. If $e$ belongs to a $D\in \cal D$, then both $\rho[i]$ and $\rho[j]$ must contain an external edge of $D$ (By Observation \ref{obs:atleast_one_ext}), and in that case, by construction of $\pi$ and definition of $\rho$, $\rho[i]$ and $\rho[j]$ must be adjacent in $\rho$, or must be terminals of $\rho$ (possible if $\pi$ corresponds to a component of $H$, which is a cycle). In other words, either $i=(j+1)\mod|\rho|$, or $j=(i+1)\mod|\rho|$. And thus, \{$\rho[i],\rho[j]\}$ is compatible to $(f_d,f_e,\alpha)$, hence they are pairwise edge disjoint, contradicting our assumption that $e$ belongs to $D$. Further, $e$ cannot belong to a $T\in \cal T$, because for every $T\in \cal T$, $T$ has only one external edge and that belongs to $A$. Thus, for every $T\in {\cal T}$, the external edge of $T$ can belongs to at most one of $\pi[i]$ and $\pi[j]$, hence it belongs to at most one of $\rho[i]$ and $\rho[j]$. This implies that at most one of the $\rho[i]$ and $\rho[j]$ can contain an edge of $T$ (By Observation \ref{obs:atleast_one_ext}). This contradicts the fact that $e$ belongs to a $T\in \cal T$, and thus contradicts our assumption that $\rho[i]$ and $\rho[j]$ intersects at an edge.
\paragraph{}

Further, $\rho$ satisfies the second condition of being compatible to $(f_d,f_e,\alpha)$, as it is a candidate of $\pi$. For the third condition, for $D\in \cal D$, let $P_D\subseteq \rho$ be the set of all the paths of $\rho$ which contain an edge of $D$. We claim that there exists a $j\in[|\rho|]$ such that indices of paths of $P_D$ in $\rho$ belong to $\{j,(j+1)\mod|\rho|\}$. This is because $D$ has only two external edges, and paths in $\rho$ are pairwise edge disjoint and are non internal to $D$, hence $|P_D|\leq 2$. Further, by construction of $\pi$ and definition of $\rho$, if $|P_D|= 2$, then both the paths of $P_D$ must be adjacent in $\rho$ or must be terminals of $\rho$ (possible if $\pi$ corresponds to a component of $H$, which is a cycle). Then there exists a $j\in[|\rho|]$, such that $D-E(set(\rho))= D-E(P_D)= D-E(\{\rho[j],\rho[(j+1)\mod|\rho|]\})$, and thus ${\textsc{opt}}(D-E(set(\rho)))={\textsc{opt}}(D-E(\{\rho[j],\rho[(j+1)\mod|\rho|]\}))\geq ({\textsc{opt}}(D)-f_d(D))$, and this holds for every $D\in \cal D$. For the fourth condition, similar arguments as above work for every $T\in \cal T$.
\end{proof}

\paragraph{}

The above lemma, helps us verify if a sequence $\pi\in \Pi$ has a candidate $\rho$ which is compatible to $(f_d,f_e,\alpha)$, by just verifying if paths for adjacent sets in $\pi$ are compatible to $(f_d,f_e,\alpha)$. If the size of $\pi$ is $\leq 2$, then by iterating over every distinct subset of $\cal P$ of size $|\pi|$, we can verify if there exists a candidate of $\pi$ which is compatible to $(f_d,f_e,\alpha)$. This will take time $(|V|+|\calp|)^{O(1)}$. We now move on to the case when $|\pi|\geq 3$.

Given a type sequence $\pi\in \Pi$, we create an auxiliary directed graph $H_{\pi}$ as follows.
\begin{itemize}
    \item  For every $j\in [|\pi|]$,  construct a vertex set $V_j=\{v_p \mid\ (p\in {\cal P}) \land (E(p)\cap {\cal A} = \pi[j]) \land \{p\}\text{ is compatible to }(f_d,f_e,\alpha)\}$. 
    
    \item For every $j\in [|\pi|]$, let $j'= (j+1) \mod |\pi|$, we add an arc (directed edge) from $v_p\in V_j$ to $v_q\in V_{j'}$ if and only if $\{p,q\}$ is compatible to $(f_d,f_e,\alpha)$. 
\end{itemize}
\paragraph{}

We note that the construction of $H_{\pi}$ takes time $(|V|+|{\cal P}|)^{O(1)}$. The constructed $H_{\pi}$ is helpful due to the following lemma.

\begin{lemma}\label{lemma:fdfe_cycle}
$\pi$ has a candidate $\rho$ such that $set(\rho)$ is compatible to $(f_d,f_e,\alpha)$ if and only if $H_{\pi}$ has a cycle containing exactly one vertex from every vertex set in $\{V_1,V_2,\dots,  V_{|\pi|}\}$.
\end{lemma}

\begin{proof}
For the forward direction, let $\rho$ be a candidate of $\pi$ such that $set(\rho)$ is compatible to $(f_d,f_e,\alpha)$. Then every subset of $set(\rho)$ is compatible to $(f_d,f_e,\alpha)$, and by construction of $H_{\pi}$, there is a directed edge in $H_{\pi}$ from the vertex corresponding to $\rho[j]$ to the vertex corresponding to $\rho[(j+1) \mod |\rho|]$ for every $j\in [|\pi|]$. Hence, there is a cycle in $H_{\pi}$ containing exactly one vertex from every vertex set in $\{V_1,V_2,\dots,  V_{|\pi|}\}$.
\paragraph{}

For the other direction, let there be such a cycle $c$ in $H_{\pi}$. As per the construction of $H_{\pi}$, every directed edge in cycle $c$ must be from a vertex of set $V_j$ to a vertex of set $V_{(j+1) \mod |\pi|}$ where $j\in [|\pi|]$. We construct $\rho$ by setting $\rho[j]$ to be the path corresponding to the vertex of $V_j$ in the cycle $c$. Thus, $\{\rho[j], \rho_i[(j+1) \mod |\pi_i|]\}$ is compatible to $(f_d,f_e,\alpha)$ as their corresponding vertices have a directed edge in $H_{\pi}$. Recalling Lemma \ref{lemma:adjacent}, we conclude that $set(\rho)$ is compatible to $(f_d,f_e,\alpha)$. This finishes the proof.
\end{proof}

\begin{lemma}\label{lemma:verify_cycle}
In time $(|V|+|{\cal P}|)^{O(1)}$, we can find a cycle containing exactly one vertex from every set in $\{V_1,V_2,\dots,  V_{|\pi|}\}$ in $H_{\pi}$ or conclude that no such cycle exists.
\end{lemma}
\begin{proof}
We claim that $H_{\pi}$ has a cycle containing exactly one vertex from every vertex set in $\{V_1,V_2,\dots,V_{|\pi|}\}$  if and only if there exist two vertices $u\in V_1$ and $v\in V_2$ such that there is a arc from $u$ to $v$ and the distance from $v$ to $u$ is $|\pi|-1$. For the forward direction, let there be such a cycle $c$ in $H_{\pi}$. Let $u$ and $v$ be the vertices in $c$ from $V_1$ and $V_2$ respectively. Since all the directed edges in $H_i$ are from vertices of vertex sets $V_j$ to the vertices of $V_{(j+1)mod|\pi_i|}$ where $j\in[|\pi_i|]$, thus there must be a directed edge from $u$ to $v$ and this directed edge must be part of cycle $c$, and the existence of cycle $c$ ensures that $u$ is reachable from $v$ by traversing $|\pi|-1$ edges. Since at least $|\pi|-1$ edges are required to reach a vertex of $V_1$ from a vertex of $V_2$. Thus, distance from $u$ to $v$ is $|\pi|-1$. For the other direction let $u\in V_1$ and $v\in V_2$ be two vertices such that $(u,v)$ is a directed edge in $H_{\pi}$ and the shortest distance from $v$ to $u$ is $|\pi|-1$. Then the shortest path from $v$ to $u$ contains exactly one vertex from every vertex  set, and since $(u,v)$ is an edge, it forms a cycle in $H_{\pi}$containing exactly one vertex from every vertex set in $\{V_1,V_2,\dots,  V_{|\pi|}\}$.
\paragraph{}

To verify if such a cycle exists, we can verify for each vertex $u\in V_1$, if there exists a vertex $v\in V_2$ such that $(u,v)$ is a directed edge and the shortest distance from $v$ to $u$ is $|\pi|-1$. Further, the shortest distance between two vertices in directed graph can be computed in time polynomial in number of vertices of the graph. This finishes the proof of the lemma. 

\end{proof}
\paragraph{}

The above lemma finishes the proof of Theorem \ref{theorem:FVD_maxdegree}.

\subsection{4-Approximation Algorithm in FPT time parameterized by Feedback Edge Number} \label{sect:FENalg}

In this section, we present a $4$-approximation algorithm for {\psp} which runs in time $({\lambda})^{O({\lambda}^2)}\cdot (|V|+|\calp|)^{O(1)}$. We first define a set of edges.

\begin{figure}[h]

    \centering
    \begin{tikzpicture}


  
    \node [ellipse, minimum height=1.5cm,minimum width= 2cm, label ={-90:\footnotesize{\text{}}}] (z) at (2*5,-2*1) {};
   
   \node[shape=circle, draw, fill = black, scale= 0.15, font=\footnotesize] ({31}) at (2*5+0.4,-2*1+0.3){};
    \node[shape=circle, draw, fill = black, scale= 0.15, font=\footnotesize] ({32}) at (2*5-0.4,-2*1+0.3){};
    \node[shape=circle, draw, fill = black, scale= 0.15, font=\footnotesize] ({33}) at (2*5+0.4,-2*1-0.3){};
    \node[shape=circle, draw, fill = black, scale= 0.15, font=\footnotesize] ({34}) at (2*5-0.4,-2*1-0.3){};
     \node[shape=circle, draw, fill = black, scale= 0.15, font=\footnotesize] ({35}) at (2*5-0.75,-2*1-0){};
     \node[shape=circle, draw, fill = black, scale= 0.15, font=\footnotesize] ({36}) at (2*5+0.75,-2*1-0){};
   
   \draw[thick,black] ({31}) to ({32});
   \draw[thick,black] ({31}) to ({33});
   \draw[thick,black] ({31}) to ({34});
   \draw[thick,black] ({31}) to ({36});
   \draw[thick,black] ({33}) to ({36});
   \draw[thick,black] ({32}) to ({33});
   \draw[thick,black] ({32}) to ({34});
   \draw[thick,black] ({33}) to ({34});
   \draw[thick,black] ({32}) to ({35});
   \draw[thick,black] ({34}) to ({35});

   \foreach \i in {1,2,3,4}{
     
            \node[shape=circle, draw, fill = black, scale= 0.15, font=\footnotesize,label={}] ({31\i}) at (2*5-0.8+\i/3,-2*1+1.2){};
   }
   \foreach \i in {1,2,3,4}{
     
            \node[shape=circle, draw, fill = black, scale= 0.15, font=\footnotesize,label={}] ({32\i}) at (2*5+1.5,-2*1-0.8+\i/3){};
   }
   
    \foreach \i in {1,2,3,4}{
     
            \node[shape=circle, draw, fill = black, scale= 0.15, font=\footnotesize,label={}] ({33\i}) at (2*5-1.5,-2*1-0.8+\i/3){};
   }
   \draw[thin,black!40] ({311}) to ({314});
   \draw[thin,black!40] ({321}) to ({324});
   \draw[thin,black!40] ({331}) to ({334});
   
   \draw[thick,black,bend right =30] ({311}) to ({32});
   \draw[thick,black,bend left =30] ({314}) to ({31});
    \draw[thick,black,bend right =30] ({324}) to ({36});
   \draw[thick,black,bend left =30] ({321}) to ({36});
    \draw[thick,black,bend right =30] ({331}) to ({35});
   \draw[thick,black,bend left =30] ({334}) to ({35});

   \node[shape=circle, draw, fill = black, scale= 0.15, font=\footnotesize,label={}] ({351}) at (2*5-0,-2*1+1.5){};
   \node[shape=circle, draw, fill = black, scale= 0.15, font=\footnotesize,label={}] ({352}) at (2*5-0.3,-2*1+1.8){};
   \node[shape=circle, draw, fill = black, scale= 0.15, font=\footnotesize,label={}] ({353}) at (2*5+0.3,-2*1+1.8){};
   
   \draw[thin,black!40] ({351}) to ({352});
   \draw[thin,black!40] ({351}) to ({353});
  \draw[thin,black!40] ({351}) to ({312});
   
   \node[shape=circle, draw, fill = black, scale= 0.15, font=\footnotesize,label={}] ({361}) at (2*5-0.5,-2*1-0.9){};
   \node[shape=circle, draw, fill = black, scale= 0.15, font=\footnotesize,label={}] ({362}) at (2*5-0.8,-2*1-1.2){};
   \node[shape=circle, draw, fill = black, scale= 0.15, font=\footnotesize,label={}] ({363}) at (2*5-0.2,-2*1-1.2){};
   
      \draw[thin,black!40] ({361}) to ({362});
   \draw[thin,black!40] ({361}) to ({363});
  \draw[thin,black!40] ({361}) to ({34});

  \node[shape=circle, draw, fill = black, scale= 0.15, font=\footnotesize,label={}] ({381}) at (2.2*5-0.5,-2*1-0.9){};
   \node[shape=circle, draw, fill = black, scale= 0.15, font=\footnotesize,label={}] ({382}) at (2.2*5-0.8,-2*1-1.2){};
   \node[shape=circle, draw, fill = black, scale= 0.15, font=\footnotesize,label={}] ({383}) at (2.2*5-0.2,-2*1-1.2){};
   
      \draw[thin,black!40] ({381}) to ({382});
   \draw[thin,black!40] ({381}) to ({383});
  \draw[thin,black!40] ({381}) to ({34});
   
    \node[shape=circle, draw, fill = black, scale= 0.15, font=\footnotesize,label={}] ({371}) at (2*5+2,-2*1+0){};
   \node[shape=circle, draw, fill = black, scale= 0.15, font=\footnotesize,label={}] ({372}) at (2*5+2.3,-2*1-0.3){};
   \node[shape=circle, draw, fill = black, scale= 0.15, font=\footnotesize,label={}] ({373}) at (2*5+2.3,-2*1+0.3){};
   
   \draw[thin,black!40] ({371}) to ({372});
   \draw[thin,black!40] ({371}) to ({373});
  \draw[thin,black!40] ({371}) to ({322});
  
  \node [ label ={-20:\footnotesize{\text{$G[X\cup S \cup T]$}}}] (z) at (1.75*5,-2.6*1-0.8) {};


   \foreach \i in {1,2,3,4}{
     
            \node[shape=circle, draw, fill = black, scale= 0.25, font=\footnotesize,label={}] ({41\i}) at (3*5-0.8+\i/3,-2*1+1.2){};
   }
   \foreach \i in {1,2,3,4}{
     
            \node[shape=circle, draw, fill = black, scale= 0.25, font=\footnotesize,label={}] ({42\i}) at (3*5+1.5,-2*1-0.8+\i/3){};
   }
   
    \foreach \i in {1,2,3,4}{
     
            \node[shape=circle, draw, fill = black, scale= 0.25, font=\footnotesize,label={}] ({43\i}) at (3*5-1.5,-2*1-0.8+\i/3){};
   }
   \draw[thick,black] ({411}) to ({414});
   \draw[thick,black] ({421}) to ({424});
   \draw[thick,black] ({431}) to ({434});

\node[shape=circle, draw, fill = black, scale= 0.25, font=\footnotesize,label={}] ({451}) at (3*5-0,-2*1+1.5){};
   \node[shape=circle, draw, fill = black, scale= 0.25, font=\footnotesize,label={}] ({452}) at (3*5-0.3,-2*1+1.8){};
   \node[shape=circle, draw, fill = black, scale= 0.25, font=\footnotesize,label={}] ({453}) at (3*5+0.3,-2*1+1.8){};
   
   \draw[thick,black] ({451}) to ({452});
   \draw[thick,black] ({451}) to ({453});
  \draw[thick,black] ({451}) to ({412});
   
   \node[shape=circle, draw, fill = black, scale= 0.15, font=\footnotesize,label={}] ({461}) at (3*5-0.5,-2*1-0.9){};
   \node[shape=circle, draw, fill = black, scale= 0.15, font=\footnotesize,label={}] ({462}) at (3*5-0.8,-2*1-1.2){};
   \node[shape=circle, draw, fill = black, scale= 0.15, font=\footnotesize,label={}] ({463}) at (3*5-0.2,-2*1-1.2){};
   
   \draw[thin,black!40] ({461}) to ({462});
   \draw[thin,black!40] ({461}) to ({463});

   \node[shape=circle, draw, fill = black, scale= 0.15, font=\footnotesize,label={}] ({481}) at (3.2*5-0.5,-2*1-0.9){};
   \node[shape=circle, draw, fill = black, scale= 0.15, font=\footnotesize,label={}] ({482}) at (3.2*5-0.8,-2*1-1.2){};
   \node[shape=circle, draw, fill = black, scale= 0.15, font=\footnotesize,label={}] ({483}) at (3.2*5-0.2,-2*1-1.2){};
   
   \draw[thin,black!40] ({481}) to ({482});
   \draw[thin,black!40] ({481}) to ({483});

    \node[shape=circle, draw, fill = black, scale= 0.25, font=\footnotesize,label={}] ({471}) at (3*5+2,-2*1+0){};
   \node[shape=circle, draw, fill = black, scale= 0.25, font=\footnotesize,label={}] ({472}) at (3*5+2.3,-2*1-0.3){};
   \node[shape=circle, draw, fill = black, scale= 0.25, font=\footnotesize,label={}] ({473}) at (3*5+2.3,-2*1+0.3){};
   
   \draw[thick,black] ({471}) to ({472});
   \draw[thick,black] ({471}) to ({473});
  \draw[thick,black] ({471}) to ({422});
  
  \node [ label ={-20:\footnotesize{\text{$G[S\cup T]$}}}] (z) at (2.8*5,-2.8*1-0.8) {};


    \node [ellipse, minimum height=1.5cm,minimum width= 2cm, label ={-90:\footnotesize{\text{}}}] (z) at (2*5,-2*1) {};
   
   \node[shape=circle, draw, fill = black, scale= 0.15, font=\footnotesize] ({51}) at (2.5*5+0.4,-6*1+0.3){};
    \node[shape=circle, draw, fill = black, scale= 0.15, font=\footnotesize] ({52}) at (2.5*5-0.4,-6*1+0.3){};
    \node[shape=circle, draw, fill = black, scale= 0.15, font=\footnotesize] ({53}) at (2.5*5+0.4,-6*1-0.3){};
    \node[shape=circle, draw, fill = black, scale= 0.15, font=\footnotesize] ({54}) at (2.5*5-0.4,-6*1-0.3){};
     \node[shape=circle, draw, fill = black, scale= 0.15, font=\footnotesize] ({55}) at (2.5*5-0.75,-6*1-0){};
     \node[shape=circle, draw, fill = black, scale= 0.15, font=\footnotesize] ({56}) at (2.5*5+0.75,-6*1-0){};

   \foreach \i in {1,2,3,4}{
     
            \node[shape=circle, draw, fill = black, scale= 0.15, font=\footnotesize,label={}] ({51\i}) at (2.5*5-0.8+\i/3,-6*1+1.2){};
   }
   \foreach \i in {1,2,3,4}{
     
            \node[shape=circle, draw, fill = black, scale= 0.15, font=\footnotesize,label={}] ({52\i}) at (2.5*5+1.5,-6*1-0.8+\i/3){};
   }
   
    \foreach \i in {1,2,3,4}{
     
            \node[shape=circle, draw, fill = black, scale= 0.15, font=\footnotesize,label={}] ({53\i}) at (2.5*5-1.5,-6*1-0.8+\i/3){};
   }
   \draw[thin,black!40] ({511}) to ({514});
   \draw[thin,black!40] ({521}) to ({524});
   \draw[thin,black!40] ({531}) to ({534});

   \node[shape=circle, draw, fill = black, scale= 0.15, font=\footnotesize,label={}] ({551}) at (2.5*5-0,-6*1+1.5){};
   \node[shape=circle, draw, fill = black, scale= 0.15, font=\footnotesize,label={}] ({552}) at (2.5*5-0.3,-6*1+1.8){};
   \node[shape=circle, draw, fill = black, scale= 0.15, font=\footnotesize,label={}] ({553}) at (2.5*5+0.3,-6*1+1.8){};
   
   \draw[thin,black!40] ({551}) to ({552});
   \draw[thin,black!40] ({551}) to ({553});
  \draw[thin,black!40] ({551}) to ({512});
   
   \node[shape=circle, draw, fill = black, scale= 0.15, font=\footnotesize,label={}] ({561}) at (2.5*5-0.5,-6*1-0.9){};
   \node[shape=circle, draw, fill = black, scale= 0.15, font=\footnotesize,label={}] ({562}) at (2.5*5-0.8,-6*1-1.2){};
   \node[shape=circle, draw, fill = black, scale= 0.15, font=\footnotesize,label={}] ({563}) at (2.5*5-0.2,-6*1-1.2){};
   
      \draw[thin,black!40] ({561}) to ({562});
   \draw[thin,black!40] ({561}) to ({563});
  \draw[thin,black!40] ({561}) to ({54});

  \node[shape=circle, draw, fill = black, scale= 0.15, font=\footnotesize,label={}] ({581}) at (2.7*5-0.5,-6*1-0.9){};
   \node[shape=circle, draw, fill = black, scale= 0.15, font=\footnotesize,label={}] ({582}) at (2.7*5-0.8,-6*1-1.2){};
   \node[shape=circle, draw, fill = black, scale= 0.15, font=\footnotesize,label={}] ({583}) at (2.7*5-0.2,-6*1-1.2){};
   
      \draw[thin,black!40] ({581}) to ({582});
   \draw[thin,black!40] ({581}) to ({583});
  \draw[thin,black!40] ({581}) to ({54});
   
    \node[shape=circle, draw, fill = black, scale= 0.15, font=\footnotesize,label={}] ({571}) at (2.5*5+2,-6*1+0){};
   \node[shape=circle, draw, fill = black, scale= 0.15, font=\footnotesize,label={}] ({572}) at (2.5*5+2.3,-6*1-0.3){};
   \node[shape=circle, draw, fill = black, scale= 0.15, font=\footnotesize,label={}] ({573}) at (2.5*5+2.3,-6*1+0.3){};
   
   \draw[thin,black!40] ({571}) to ({572});
   \draw[thin,black!40] ({571}) to ({573});
  \draw[thin,black!40] ({571}) to ({522});
  
  \node [ label ={-20:\footnotesize{\text{$G-\cal B$}}}] (z) at (2.35*5,-6.6*1-0.8) {};

\end{tikzpicture}
\caption{The darkened edges in $G[X\cup S\cup T]$ represents the edge set ${\cal B}=E(G[X]) \cup E(X,S)$. The darkened components in $G[S\cup T]$ forms the set $\cal D$, and the lighter components forms $\cal T$. The graph $G-\cal B$ is a forest. }\label{fig:structures2}
\end{figure}
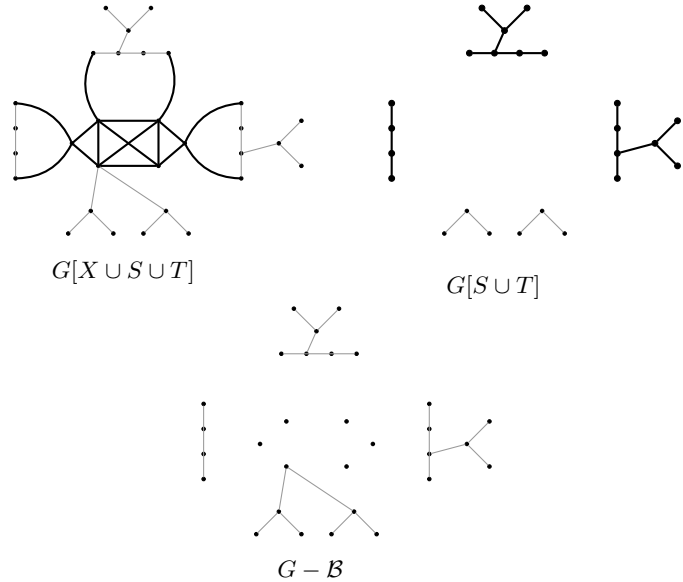

\begin{definition}
  We define ${\cal B}= E(G[X]) \cup E(X,S)$.
\end{definition}

\paragraph{}

Every edge in $E(X,S)$ is an external edge of a component in $\cal D$, recalling the Corollary \ref{cor:sizeX}, $|X|=O(|\lambda|)$ and $|{\cal D}|= O(|\lambda|)$. Further, recalling Observation \ref{obs:twoedges}, every component in $\cal D$ has two external edges. Thus, the size of $|{\cal B}|= O({\lambda}^2)$.
\paragraph{}

 Let $M\subseteq \cal P$ be a path set packing of maximum size in $\cal P$. Let $M_1 = \{p\mid p\in M\ \land \ E(p)\cap {\cal B} \neq \emptyset \}$ and $M_2 = M\setminus M_1$, that is $M_1$ is the set of all the paths of $M$ which contain at least one edge from $\cal B$, and $M_2$ are the remaining paths of $M$. We have $|M|=|M_1|+|M_2|$. Since all the edges of $M_2$ belong to the subgraph $G-{\cal B}$, we have $|M_2|\leq \textsc{opt}(G-{\cal B})$. Thus,

\begin{equation}
    |M|\leq|M_1|+\textsc{opt}(G-{\cal B}).
\end{equation}

 Let $M'_1\subseteq \cal{P}$ be a path set packing such that $|M'_1|\geq {1\over3} \cdot|M_1|$, then 

 \begin{equation}\label{eqn:4approx}
     \max(\textsc{opt}(G-{\cal B}),M'_1)\geq {1\over4}\cdot |M|.
 \end{equation}

Thus, it will suffice to find a path set packing of size $\textsc{opt}(G-{\cal B})$ and $M'_1$ in time $(\lambda)^{O({\lambda}^2)}\cdot (|V|+|{\cal P}|)^{O(1)}$. First, to find a path set packing of size $\textsc{opt}(G-{\cal B})$, consider the following lemma.

\begin{lemma}
    $G-{\cal B}$ is a forest.
\end{lemma}
\begin{proof}
    $G-{\cal B}$ is the graph with vertex set $V(G)$ and the edge set $E(G[S\cup T]) \cup E(X,T)$ (refer Figure  \ref{fig:structures2} for overview). Let $G'$ be the graph with vertex set $G[V]$ and the edge set $E(G[S\cup T])$, $G'$ is a forest in which every component either belongs to $\cal D \cup \cal T$ or is an isolated vertex that belongs to the vertex set $X$. Further, $E(X,T)$ is the set of all the external edges of components in $\cal T$. Every component in $\cal T$ is a tree and has only one external edge, thus every edge in $E(X,T)$ connects a vertex in $X$ to a distinct component in $\cal T$. We add edges of $E(X,T)$ in an arbitrary order to the graph $G'$, every time we are adding an edge between two distinct components, thus no new cycle is introduced in the graph. Since $G'$ is a forest, the obtained graph $G-\cal B$ is a forest. 
\end{proof}

\paragraph{}

Recalling Corollary \ref{cor:optforest}, a path set packing of size $\textsc{opt}(G-{\cal B})$ in $\textsc{int}(G-{\cal B},\cal P)$ can be found in time $(|V|+|{\cal{P}}|)^{O(1)}$. It remains to show that we can find $M'_1$ in time $(\lambda)^{O({\lambda}^2)}\cdot (|V|+|{\cal P}|)^{O(1)}$. To find $M'_1$, we consider $M_1$. Since $M_1$ is a path set packing and each of its path intersects with ${\cal B}$, we have $|M_1|\leq |\cal B|$. To this end, we will now try to guess $|M_1|$ and $E(p)\cap {\cal B}$ for every $p\in M_1$. For which, let $f_b= \{B_1,B_2,\dots,B_{|f_b|}\}$ be a partition of $\cal B$. We say $f_b$ is a \textit{correct guess }for $M_1$ if the following holds. 
\begin{enumerate}
    \item $|M_1|=|f_b|-1$;
    \item For every $p\in M_1$, there exists a $B\in f_b$, such that $E(p)\cap {\cal B}=B$.
\end{enumerate}

\begin{lemma} \label{lemma:existcorrectgusee}
     There exists a partition $f^*_b$ of $\cal B$, such that $f^*_b$ is a correct guess for $M_1$.
\end{lemma}
\begin{proof}
      $M_1$ is a path set packing and every path of $M_1$ intersects with $\cal B$. We construct $f^*_b= \{E(p)\cap {\cal B}\mid p\in M_1\} \cup \{
 {\cal B}\setminus E(M_1)\}$. Note that $\{{\cal B}\setminus E(M_1)\}$ is nonempty, as $\cal B$ contains edges which belong to no path in $\cal P$ (at least the edges between the vertices of $\{z_1,z_2,z_3\}$ which we added separately). Further, $E(p)\cap {\cal B}$ is distinct, disjoint, and nonempty for every $p\in M_1$, thus $f^*_b$ is a partition of $\cal B$ and $|M_1|=|f^*_b|-1$.
\end{proof}

\begin{definition}
    For an edge set $F\subseteq E$, we define $\Upsilon(F)= \{C\mid C\in ({\cal D}\cup {\cal T}) \land |\textsc{ext\_e}(C) \cap F|=1\}$, that is the set of all the components of ${\cal D}\cup {\cal T}$ whose exactly one external edge belong to $F$.
\end{definition}

\paragraph{}

Given a partition $f_b$ of $\cal B$, we proceed as follows.
Let ${\cal P}_{1}=\{p\mid (p\in{\cal P}) \land (B\in f_b) \land (E(p)\cap {\cal B} = B) \}$. 
For a universe $U= f_b \cup {\cal{D}}\cup {\cal{T}}$, we construct a collection of sets $\cal Q$ as follows. \newline
 
\begin{itemize}
        \item For every $p\in {\cal P}_{1}$, we define $Q(p)= \{B\} \cup \Upsilon(E(p))$, where $B=E(p)\cap \cal B$.
        \item ${\cal Q}= \{Q(p)\mid p\in {\cal P}_1\}$.
    \end{itemize} 

\begin{claim}
    For every $p\in {\cal P}_1$, $|Q(p)|\leq 3$.
\end{claim}

\begin{proof}
   This follows from the fact that for every $p\in {\cal P}_{1}$, $|\Upsilon(E(p))|\leq 2$ ( Lemma \ref{lemma:paththreeormore1}).
   
\end{proof}

\paragraph{}

For every $l\in[|f_b|-1]$, we construct an instances $I=(U,{\cal Q},l)$ of $3$-{\sp}, and solve it using Theorem \ref{thm:3setpacking}. Let $l^*$ be the maximum integer in $[|f_b|-1]$ such that $I=(U,{\cal Q},l^*)$ is a yes instance, and let ${\cal{Q}}^*$ be the set packing in $\cal Q$ obtained for this instance. We then construct a set of paths $M^*$ as follows.

\begin{enumerate}
        \item Initialize $M^*=\emptyset.$
        \item For every $q\in Q^*$:
        \begin{itemize}
            \item Arbitrarily pick a path $p\in {\cal P}_{1}$, such that $Q(p)=q$, and add $p$ to  $M^*$.       
        \end{itemize} 
\end{enumerate}
\paragraph{}

Since we are solving $O(|\cal B|)$ instances of $3$-{\sp} by Theorem \ref{thm:3setpacking}, and each takes time $2^{O(|{\cal B}|)}\cdot (|V|+|{\cal P}|)^{O(1)}$.
Thus, for a given $f_b$, the construction of $M^*$ takes time $2^{O(|{\cal B}|)}\cdot (|V|+|{\cal P}|)^{O(1)}$. Further, to bound the size of $M^*$, consider the following.

\begin{lemma}\label{lemma:correctguess}
    If $f_b$ is a correct guess for $M_1$, then $|M^*|\geq {1\over 3}\cdot |M_{1}|$.
\end{lemma}
\begin{proof}
   We are adding a path in $M^*$ for every $q\in Q^*$, hence $|{\cal Q}^*|= |M^*|$. It will now suffice to prove that $|Q^*|\geq {1\over 3}\cdot |M_{1}|$. To prove our claim, we construct two sets of paths $M_{1,1}$ and $M_{1,2}$ as follows.
    \begin{enumerate}
      
        \item Initialize $M_{1,1}=M_{1,2}=\emptyset.$
        \item While $M_{1}\neq \emptyset$:
        \begin{enumerate}
            \item Arbitrarily pick a path $p\in M_{1}$, and let $N(p)=\{p' \mid p' \in (M_{1}\setminus \{p\})\land (\Upsilon(E(p))\cap \Upsilon(E(p'))\neq \emptyset) \}$.           
            \item Set $M_{1,1}=M_{1,1}\cup \{p\}$, $M_{1,2}=M_{1,2}\cup N(p)$, and $M_{1}=M_{1}\setminus (\{p\}\cup N(p))$. 
        \end{enumerate} 
    \end{enumerate}
    \paragraph{}

    We claim that for every $p\in M_{1}$, $N(p)=|\Upsilon(E(p))|$. This is because $M_1$ is a path set packing, hence for every $C\in \Upsilon(E(p))$, at most two distinct paths in $M_{1}$ can have its external edge, and one of them is $p$ itself. Thus, there are at most $|\Upsilon(E(p))|$ paths $p'\in (M_{1}\setminus \{p\})$ such that $\Upsilon(E(p))\cap \Upsilon(E(p'))\neq \emptyset$. Recalling Lemma \ref{lemma:paththreeormore1}, we have $|N(p)|\leq 2$. Further, in the construction above, each time we are adding a path $p$ in $M_{1,1}$, we are adding at most two paths $N(p)$ in $M_{1,2}$. Thus, $|M_{1,1}|\geq {1\over 2}\cdot|M_{1,2}|$, and since $|M_{1,1}|+|M_{1,2}|= |M_{1}|$, we have  $|M_{1,1}|\geq {1\over{3}}\cdot |M_{1}|$.
    \paragraph{}

    If $f_b$ is a correct guess for $M_1$, then for every $p\in M_1$, $E(p)\cap {\cal B}= B$ for a $B\in f_b$. Thus, $M_1\subseteq {\cal P}_1$, and $Q(p)$ is well defined for every $p\in M_1$. Let ${\cal Q}'=\{Q(p)\mid p\in M_{1,1}\}$. Since $M_{1,1}$ is a path set packing, $E(p)\cap \cal B$ is distinct for every $p\in M_{1,1}$, hence $Q(p)$ is distinct for every $p\in M_{1,1}$, and hence $|{\cal Q}'|= |M_{1,1}|$. Further, ${\cal Q}'$ is a set packing in $Q$ because for every two distinct paths $p_1,p_2\in M_{1,1}$, $\Upsilon(E(p_1))\cap \Upsilon(E(p_2)) =\emptyset$. Further, due to the maximality of ${\cal Q}^*$, we have $|{\cal Q}^*|\geq|{\cal Q}'|=|M_{1,1}|\geq {1\over{3}}\cdot |M_{1}|$.
\end{proof}

\begin{lemma}\label{lemma:correctguesspacking}
    $M^*$ is a path set packing in $\cal P$.
\end{lemma}

\begin{proof}

    Assume to the contrary that there exist two distinct paths $p_1,p_2\in {M^*}$ such that $E(p_1)\cap E(p_2)\neq \emptyset$. 
    Let $E(p_1)\cap {\cal B}=B_1$ and $E(p_2)\cap {\cal B}=B_2$. Since $B_1\in Q(p_1)$ and $B_2\in Q(p_2)$, $B_1$ and $B_2$ are distinct and disjoint, and no edge in $\cal B$ belongs to $E(p_1)\cap E(p_2)$. Since $E(G[X])\subseteq \cal B$, no edge in $E(G[X])$ belongs to $E(p_1)\cap E(p_2)$.
    \paragraph{}

    External edges of every component in $\cal D$ forms a subset of   $\cal B$, and as discussed above, no edge in $\cal B$ belongs to $E(p_1)\cap E(p_2)$. Further, every $T\in {\cal T}$ has exactly one external edge, and if the external edge of a $T\in {\cal T}$ belongs to both $p_1$ and $p_2$, then both $\Upsilon(E(p_1))$ and $\Upsilon(E(p_2))$ must contain $T$, which is not the case as $Q(p_1)$ and $Q(p_2)$ are disjoint, and hence $\Upsilon(E(p_1))$ and $\Upsilon(E(p_2))$ must be disjoint as well. Recalling that ${\cal D} \cup {\cal T}$ is the set of all the components in $G[S\cup T]$, and union of external edges of all the components in ${\cal D} \cup {\cal T}$ is $E[X,S\cup T]$, no edge in $E[X,S\cup T]$ belongs to $E(p_1)\cap E(p_2)$.
    \paragraph{}

    As discussed above, for every $T\in \cal T$, at most one of the $p_1$ and $p_2$ contains the external edge of $T$, recalling Observation \ref{obs:atleast_one_ext}, this implies that at most one of the $E(p_1)$ and $E(p_2)$ contains an edge of $T$. Further, we argue that for every $D \in \cal D$, at most one of the $p_1$ and $p_2$ contains an external edge of $D$. This is because no external edge of $D$ belongs to $E(p_1)\cap E(p_2)$, and if both $p_1$ and $p_2$ contain an external edge of $D$, then both of them contain exactly one external edge of $D$, and in this case, $D$ must be in both $\Upsilon(E(p_1))$ and $\Upsilon(E(p_2))$, which is not the case. This implies that at most one of the $E(p_1)$ and $E(p_2)$ contains an edge of $D$.
    Since ${\cal D}\cup {\cal T}$ is the set of all the components in $G[S\cup T]$, no edge in $E(G[S\cup T])$ belongs to $E(p_1)\cap E(p_2)$. Since $E(G[X])\cup E(X, S\cup T) \cup E(G[S\cup T]) = E(G)$, it contradicts our assumption that $E(p_1)\cap E(p_2)\neq \emptyset$.
\end{proof}

\paragraph{}

We construct $M^*$ for every possible $f_b$, and pick the maximum sized $M^*$ to be $M'_1$, by Lemma \ref{lemma:existcorrectgusee}, Lemma \ref{lemma:correctguess} and Lemma \ref{lemma:correctguesspacking}, $|M'_1|\geq {1\over 3}\cdot |M_1|$. This will allow us to use Equation \ref{eqn:4approx} to find a path set packing of size at least ${1\over4}\cdot |M|$.
 \paragraph{}

Since there are $\lambda^{O({\lambda}^2)}$ possible partitions $f_b$ of $\cal B$, and as discussed, we can construct $M^*$ for every $f_b$ in 
time $2^{O(|{\lambda}|)}\cdot (|V|+|{\cal P}|)^{O(1)}$. Finding $\textsc{opt}(G-\cal B)$ takes time $(|V|+|{\cal P}|)^{O(1)}$. Thus, the algorithm takes time $\lambda^{O({\lambda}^2)}\cdot (|V|+|{\cal P}|)^{O(1)}$. This finishes the proof of Theorem \ref{theorem:4-apxFEN}.

\subsection{FPT Algorithm Parameterized by Treewidth + Maximum Degree + Maximum Path Length} \label{sec:3parameters}

Given a pair $(G,\calp)$, consider the auxiliary graph $H$ whose vertices correspond to paths in $\calp$ and in which two vertices are adjacent if the corresponding paths intersect at an edge. Then finding the maximum path set packing in $\calp$ is equivalent to finding a maximum size independent set in $H$.
We can deduce the following facts about the structure of $H$.

\begin{lemma} \label{lemma:fpt}
If $G$ has treewidth $\tau$ and maximum degree $\Delta$, and each path in $\calp$ is of length at most $r$, then the treewidth of $H$ is at most $(\tau+1)\Delta^r$.
\end{lemma}

\begin{proof}
Given a tree decomposition $(T,\beta)$ of $G$ of width $\tau$, consider the pair $(T,\gamma)$, where $\gamma(v)=\cup_{x \in \beta(v)} \cup \{p\}$,
where the inner union is over all paths $p$ in $G$ of length at most $r$ and containing $x$.
We have $|\gamma(v)| \leq (\tau+1)\Delta^{r}$ since each vertex in $\beta(v)$ belongs to at most $\Delta^{r}$ paths of length at most $r$ in $G$. Further $(T,\gamma)$ is a tree decomposition of $H$: 

(a) for every $p\in \calp$, there exists a node $v$ in $T$ such that $p$ is in $\gamma(v)$, eg: node $v\in V(T)$ such that $\beta(v)$ contains a vertex of $p$;

(b) if $p,q\in \calp$ are adjacent in $H$ that is $p,q$ intersect at an edge, then both $p,q$ belong to $\gamma(v)$ for a node $v$ in $T$ such that $\beta(v)$ contains a common vertex of $p,q$;

(c) now we show that for every $p \in \calp$, the set of all the nodes $v\in V(T)$ such that $\gamma(v)$ contains $p$ induces a connected subgraph of $T$.
We assume to the contrary that there is a path $p \in \calp$ such that the set $S(p)$ of all the nodes in $T$ such that $\gamma(v)$ contains $p$ does not induce a connected subgraph of $T$, then there must be a node $v$ in $T$ such that $p \notin \gamma(v)$ and $v$ belongs to the path between two distinct nodes in $S(p)$ which belong to separate components in $T[S(p)]$. Graph $G-\beta(v)$ is disconnected
and has at least two connected components containing two distinct vertices, say $x,y$ of the path $p$. This implies that at least one vertex of the subpath of $p$ joining $x,y$ must belong to $\beta(v)$, and thus $p$ belongs to $\gamma(v)$, a contraposition to our assumption.
\end{proof}

The problem of finding maximum independent set admits a $O(2^{\tau})$ algorithm on graphs of treewidth $\tau$ \cite{AP89}.
Hence we obtain Theorem \ref{theorem:tw_maxdegree} as a corollary.

\section{Conclusion}

We investigated parameterized complexity of {\psp} with various structural and natural parameters. The hardness and FPT results that we obtained help establish parameterized complexity of {\psp} with respect to most widely used parameters. We leave open the question if {\psp} is FPT parameterized by the feedback edge number of the input graph.

\paragraph{}
\textbf{Acknowledgements.}
We thank anonymous reviewer of an earlier version of this paper for referring us to EPG and EPT graphs. We also thank reviewers of WALCOM 2023 for their valuable suggestions and feedback.

\bibliographystyle{splncs04}

\bibliography{path-set-references}

\end{document}